\RequirePackage[l2tabu,orthodox]{nag}
\documentclass[final,12pt,a4paper,abstracton]{scrartcl}

\usepackage{authblk}

\usepackage[utf8]{inputenc}
\usepackage[T1]{fontenc}
\usepackage{lmodern}
\usepackage{microtype}

\usepackage{amssymb}
\usepackage{amsmath}
\usepackage{bm}
\usepackage{mathtools}

\usepackage{enumitem}
\setlist[itemize]{label=$\circ$}
\setlist[description]{labelindent=\parindent,font=\normalfont\itshape}

\usepackage{tikz}
\usepackage{ifthen}
\usetikzlibrary{calc}

\usepackage[
  authordate,
  uniquelist=false,
  noibid,
  useprefix=true,
  backend=biber
]{biblatex-chicago}
\renewcommand\UrlFont{\sffamily\footnotesize}
\renewcommand\path[1]{{\sffamily\footnotesize\detokenize{#1}}}
\bibliography{references}

\usepackage[amsmath,hyperref,thmmarks]{ntheorem}

\theoremseparator{.}
\newtheorem{theorem}{Theorem}[section]
\newtheorem{lemma}[theorem]{Lemma}
\newtheorem{definition}[theorem]{Definition}
\newtheorem{proposition}[theorem]{Proposition}

\theoremstyle{nonumberplain}
\theoremheaderfont{\normalfont\itshape}
\theorembodyfont{\normalfont}
\theoremsymbol{\ensuremath{_\blacksquare}}
\newtheorem{proof}{Proof}

\usepackage{hyperref}
\usepackage{xcolor}
\hypersetup{
    colorlinks,
    linkcolor={blue!50!black},
    citecolor={blue!50!black},
    urlcolor={black}
}

\DeclarePairedDelimiter\paren{\lparen}{\rparen}
\DeclarePairedDelimiter\abs{\lvert}{\rvert}
\DeclarePairedDelimiter\set{\{}{\}}
\DeclarePairedDelimiterX\setc[2]{\{}{\}}{\,#1 \;\delimsize\vert\; #2\,}
\DeclarePairedDelimiterX\Prc[2]{\Pr\lparen}{\rparen}{\,#1 \;\delimsize\vert\; #2\,}

\newcommand{\cc}[1]{\ensuremath{\mathsf{#1}}}
\newcommand{\pp}[1]{\textsc{#1}}
\newcommand{\op}[1]{\ensuremath{\operatorname{#1}}}

\renewcommand{\P}{\cc{P}}
\newcommand{\NP}{\cc{NP}}
\newcommand{\W}{\cc{W}}
\newcommand{\FPT}{\cc{FPT}}

\newcommand{\poly}{\op{poly}}

\newcommand{\N}{\mathbb{N}}
\newcommand{\Z}{\mathbb{Z}}
\newcommand{\Q}{\mathbb{Q}}
\newcommand{\R}{\mathbb{R}}

\newcommand{\dotcup}{\mathbin{\dot\cup}}

\newcommand{\paths}{\op{Paths}}

\newcommand{\hypo}{\ensuremath{ \cc{coNP} \not\subseteq \cc{NP/poly}}}
\newcommand{\hypofails}{\ensuremath{\cc{coNP} \subseteq \cc{NP/poly}}}

\newcommand{\OR}{\op{OR}}

\newcommand{\setmatching}[1]{\pp{\mbox{$#1$-Set} Matching}}
\newcommand{\perfectsetmatching}[1]{\pp{Perfect \mbox{$#1$-Set} 
    Matching}}
\newcommand{\perfectdimensionalsetmatching}[1]{\pp{Perfect 
    \mbox{$#1$-Di}\-men\-sion\-al Matching}}
\newcommand{\factor}[1]{\pp{$#1$-Factor}}
\newcommand{\matching}[1]{\pp{$#1$-Matching}}

\let\epsilon\varepsilon

\begin{document}

\title{Kernelization of Packing Problems}
\author[1]{Holger Dell}%
\author[2]{Dániel Marx}%

\affil[1]{%
	Saarland Informatics Campus, Saarbrücken, Germany%
}

\affil[2]{%
	Institute for Computer Science and Control, Hungarian Academy of Sciences (MTA SZTAKI), Hungary%
}

\date{November 26, 2018\footnote{HD was partially supported by NSF grant 1017597 and by the
Alexander von Humboldt Foundation. DM received funding from the European Research Council (ERC) under the European Union's Horizon 2020 research and innovation programme
under grant agreements No.~280152, 725978. An extended abstract of this manuscript was presented at SODA 2012 (\href{https://dx.doi.org/10.1137/1.9781611973099.6}{doi:10.1137/1.9781611973099.6}).}}

\maketitle

\begin{abstract}
	Kernelization algorithms are polynomial-time reductions from a problem
	to itself that guarantee their output to have a size not exceeding
	some bound.  For example, \setmatching{d} for integers $d\geq 3$
	is the problem of finding a matching of size at least $k$ in a given
	$d$-uniform hypergraph and has kernels with $O(k^d)$ edges.
	Bodlaender et al.\ [JCSS 2009], Fortnow and Santhanam [JCSS 2011],
	Dell and Van Melkebeek [JACM 2014] developed a framework for proving
	lower bounds on the kernel size for certain problems, under the
	complexity-theoretic hypothesis that \cc{coNP} is not contained in \cc{NP/poly}.
	Under the same hypothesis, we show lower bounds for the kernelization
	of \setmatching{d} and other packing problems.

	Our bounds are tight for \setmatching{d}: It does not have
	kernels with $O(k^{d-\epsilon})$ edges for any $\epsilon>0$ unless the
	hypothesis fails.  By reduction, this transfers to a bound of
	$O(k^{d-1-\epsilon})$ for the problem of finding~$k$ vertex-disjoint
	cliques of size~$d$ in standard graphs.  Obtaining tight bounds for graph packing problems is challenging: We make first progress in this direction by showing non-trivial kernels with~$O(k^{2.5})$ edges for the problem of finding $k$ vertex-disjoint
	paths of three edges each. If the paths have~$d$ edges each, we improve the straightforward~$O(k^{d+1})$ kernel can be improved to a uniform polynomial kernel where the exponent of the kernel size is independent of~$k$.

	Most of our lower bound proofs follow a general scheme that we
	discover: To exclude kernels of size~$O(k^{d-\epsilon})$ for a problem
	in $d$-uniform hypergraphs, one should reduce from a carefully chosen
	$d$-partite problem that is still \NP-hard.  As an illustration, we
	apply this scheme to the vertex cover problem, which allows us to
	replace the number-theoretical construction by Dell and Van Melkebeek
		[JACM 2014] with shorter elementary arguments.
\end{abstract}

\section{Introduction}
Algorithms based on kernelization play a central role in
fixed-parameter tractability and perhaps this kind of parameterized
algorithms has the most relevance to practical computing.
Recall that a problem is {\em fixed-parameter tractable}
parameterized by some parameter~$k$ of the instance if it can be
solved in time \mbox{$f(k)\cdot n^{O(1)}$} for some computable function $f$
depending only on the parameter $k$ (see
\textcite{DBLP:books/sp/CyganFKLMPPS15,MR2001b:68042,FG06,MR2223196}).
A {\em kernelization algorithm}
for a problem $P$ is a polynomial-time algorithm that, given an instance
$x$ of the problem $P$ with parameter $k$, creates an equivalent instance~$x'$
of $P$ such that the size of~$x'$ and the parameter~$k'$ of the new instance are both bounded from above by a function
$f(k)$.
For example, the classical result of \textcite{NemhauserTrotter} can be
interpreted as a kernelization
algorithm that, given an instance of \pp{Vertex Cover}, produces an
equivalent instance on at most $2k$ vertices, which implies that it
has at most~$\binom{2k}{2}$ edges.  A kernelization algorithm can be
thought of as preprocessing that creates an equivalent instance whose
size has a mathematically provable upper bound that depends only on
the parameter of the original instance and not on the size of the
original instance.  Practical computing often consists of a heuristic
preprocessing phase to simplify the instance followed by an exhaustive
search for solutions (by whatever method available).  Clearly, it is
desirable that the preprocessing shrinks the size of the instance as
much as possible.  Kernelization is a framework in which the
efficiency of the preprocessing can be studied in a rigorous way.

One can find several examples in the parameterized complexity
literature for problems that admit a kernel with relatively small
sizes, i.e., for problems where $f(k)$ is polynomial in $k$.  There
are efficient techniques for obtaining such results for particular
problems (e.g.,
\textcite{ThomasseFVS,DBLP:journals/tcs/Guo09,DBLP:conf/wg/ChorFJ04,DBLP:journals/tcs/LokshtanovMS11,DBLP:journals/talg/Binkele-RaibleFFLSV12}).
Some
of these techniques go back to the early days of parameterized
complexity and have been refined for several years. In later developments,
general abstract techniques were developed that give us kernelization
results for several problems at once
\parencite{DBLP:conf/soda/FominLST10,DBLP:journals/jacm/BodlaenderFLPST16}.

\textcite{BDFH} developed a framework for
showing that certain parameterized problems are unlikely to have
kernels of polynomial size, and \textcite{FortnowSanthanam} proved the
connection with the complexity-theoretic hypothesis \hypo{}.
In particular, for several basic problems, such as finding a cycle of length $k$, a kernelization
with polynomial size would imply that the hypothesis is false.
The framework of \textcite{BDFH} has led to a long series of hardness
results showing that several concrete problems with various
parameterizations are unlikely to have kernels of polynomial size
\parencite{ChenFlumMuller07,DBLP:journals/tcs/BodlaenderTY11,DBLP:journals/talg/DomLS14,DBLP:journals/talg/Binkele-RaibleFFLSV12,KratschWahlstrom09a,DBLP:journals/disopt/KratschW13,DBLP:journals/toct/KratschMW16,DBLP:journals/siamdm/BodlaenderJK13,DBLP:journals/siamdm/BodlaenderJK14,DBLP:journals/disopt/MisraRS11}.

\textcite{DellVanMelkebeek} refined
the complexity results of~\textcite{FortnowSanthanam,BDFH} to prove
conditional lower bounds also for problems that do admit polynomial
kernels.
For example, they show that \pp{Vertex Cover} does not have kernels of
size $O(k^{2-\epsilon})$ unless the hypothesis \hypo{} fails.
Similar lower bounds are given for several other graph covering
problems where the goal is to delete the minimum number of vertices in
such a way that the remaining graph satisfies some prescribed
property.
Many of the lower bounds are tight as they match the upper bounds of
the best known kernelization algorithms up to an arbitrarily
small~$\epsilon$ term in the exponent.

In the present paper, we also obtain kernel lower bounds for problems
that have polynomial kernels, but the family of problems that we
investigate is very different: packing problems.
Covering and packing problems are dual to each other, but there are
significant differences in the way they behave with respect to
fixed-parameter tractability.
For example, techniques such as bounded search trees or iterative
compression are mostly specific to covering problems, while techniques
such as color coding are mostly specific to packing problems.
\pp{Feedback Vertex Set} is the problem of covering all cycles and has
kernels with $O(k^2)$ edges~\parencite{ThomasseFVS,DellVanMelkebeek}, while
its packing version is the problem of finding~$k$ vertex-disjoint
cycles and is unlikely to have polynomial
kernels~\parencite{DBLP:journals/tcs/BodlaenderTY11}.
Therefore, the techniques for understanding the kernelization
complexity of covering and packing problems are expected to differ
very much.
Indeed, the proofs of \textcite{DellVanMelkebeek} for the problem of
covering sets of size $d$ cannot be straightforwardly adapted to the
analogous problem of packing sets of size $d$.

Our contributions are twofold. First, we obtain lower bounds on the
kernel size for packing sets and packing disjoint copies of a
prescribed subgraph $H$. An example of the latter is the problem of
finding $k$ vertex-disjoint $d$-cliques in a given graph.
For packing sets, our
lower bound is tight, while determining the best possible kernel size
for graph packing problems with every fixed $H$ remains an interesting
open question.
Fully resolving this question would most certainly involve
significantly new techniques both on the complexity and the
algorithmic side.
To indicate what kind of difficulties we need to overcome for the
resolution of this question, we show kernels with~$O(k^{2.5})$ edges
for the problem of packing~$k$ vertex-disjoint paths on three edges.
Moreover, we show that the problem of packing~$k$ vertex-disjoint paths on $d$ edges admits a uniformly polynomial kernel, that is, a polynomial kernel where the exponent of the running time and the kernel size does not depend on $d$.

Secondly, the techniques used in our lower bounds are perhaps as
important as the concrete results themselves. We present a simple and
clean way of obtaining lower bounds of the form
$O(k^{d-\epsilon})$. Roughly speaking, the idea is to reduce from an
appropriate \mbox{$d$-partite} problem by observing that if we
increase the size of the universe by a factor of~$t^{1/d}$, then we
can conveniently pack together $t$ instances.
A similar effect was achieved by~\textcite{DellVanMelkebeek}, but it used a
combinatorial tool called the \emph{Packing Lemma}, whose proof uses nontrivial
number-theoretical arguments.
As a demonstration, we show that our scheme allows us to obtain the main
kernelization results of~\textcite{DellVanMelkebeek} with very simple elementary
techniques.
Furthermore, this scheme proves to be very useful for packing
problems, even though in one of our lower bounds it was easier to
invoke the Packing Lemma.
It seems that both techniques will be needed for a complete
understanding of graph packing problems.

\subsection{Results}

The matching problem in $d$-uniform hypergraphs,
\setmatching{d}, is to decide whether a given hypergraph has a
matching of size~$k$, i.e., a set of~$k$ pairwise disjoint hyperedges.
Correspondingly, the \perfectsetmatching{d} problem is to find
a \emph{perfect} matching, i.e., a matching with $k=n/d$ where~$n$ is
the number of vertices.
Finally, \perfectdimensionalsetmatching{d} is the problem
\perfectsetmatching{d} when restricted to $d$-partite hypergraphs,
i.e., hypergraphs partitioned into $d$ color classes so that each
hyperedge has at most one vertex per color class.
\textcite{FellowsEtal} give a kernelization algorithm for
\setmatching{d}, which is the most general of these three problems.
\begin{theorem}[\cite{FellowsEtal}]\label{thm:setmatching-kernels}
	The problem \setmatching{d} has kernels with $O(k^d)$ hyperedges.
\end{theorem}
In Appendix~\ref{appendix sunflower setmatching}, we sketch a
straightforward but instructive proof of this fact using the sunflower
lemma of \textcite{ER60sunflower}.
Our main result is that the kernel size above is asymptotically
optimal under the hypothesis \hypo.
\begin{theorem} \label{thm:setpacking}
	Let $d \geq 3$ be an integer and let $\epsilon$ be a positive real.
	If~\hypo{} holds, then \perfectdimensionalsetmatching{d} does not have
	kernels with $O(k^{d-\epsilon})$ hyperedges.
\end{theorem}
Since \perfectdimensionalsetmatching{d} is the special case of
\setmatching{d} where $k=n/d$ and the input graphs are $d$-partite,
the lower bound applies to that problem as well and it
shows that the upper bound in Theorem~\ref{thm:setmatching-kernels} is
asymptotically tight.

A particularly well-studied special case of set matching is when the
sets are certain fixed subgraphs (e.g., triangles, cliques, stars,
etc.) of a given graph.  We use the terminology of
\textcite{Yuster}, who surveys graph theoretical properties of such
graph packing problems.  Formally, an {\em $H$-matching} of size $k$
in a graph $G$ is a collection of $k$ vertex-disjoint subgraphs of $G$
that are isomorphic to $H$.  The problem \pp{$H$-Matching} is to find
an $H$-matching of a given size in a given graph. Both problems are
\NP-complete whenever~$H$ contains a connected component with more
than two vertices~\parencite{KirkpatrickHell} and is in $\P$ otherwise.

The kernelization properties of graph packing problems received a lot
of attention in the literature (e.g.,~\cite{%
	moser_packing,
	FellowsTriangles,
	PietroSloper,FernauRaible,WangEtal,
	MathiesonEtal
}).
The problem \matching{H} can be expressed as a \setmatching{d}
instance with $O(k^{d})$ edges (where $d\coloneqq|V(H)|$) and therefore
Theorem~\ref{thm:setmatching-kernels} implies a kernel of size~$O(k^{d})$.
In the particularly interesting special case when $H$ is
a clique $K_d$, we use a simple reduction to transfer the above
theorem to obtain a lower bound for \pp{$K_d$-Matching}.
\begin{theorem} \label{thm:cliquepacking}
  Let $d\geq 3$ be an integer and let $\epsilon$ be a positive real.
  If \hypo{} holds, then \matching{K_d} does not have kernels with $O(k^{d-1-\epsilon})$ edges.
\end{theorem}
An upper bound of size $O(k^{d})$ follows for \matching{K_d} from
Theorem~\ref{thm:setmatching-kernels}.  This does not quite match our
conditional lower bounds of $O(k^{d-1-\epsilon})$, and it is an
interesting open problem to make the bounds tight.

The \factor{H} problem is the restriction of \matching{H} to
the case $k=n/d$, that is, the task is to find an $H$-matching that
involves all vertices.
For every fixed graph~$H$, the problem \factor{H} has kernels with $O(k^2)$ edges, for the trivial reason that an $n$-vertex instance has size $O(n^2)$ and we have $k=\Theta(n)$ by the definition of \factor{H}.
We show that this bound is tight for every \NP-hard \factor{H} problem.
\begin{theorem} \label{thm:Hfactor}
	Let $H$ be a connected graph with $d \geq 3$ vertices and $\epsilon$
  a positive real.
  If \hypo{} holds, then \factor{H} does not have kernels with $O(k^{2-\epsilon})$ edges.
\end{theorem}
Thus, it is unlikely that instances of \factor{H} can be reduced in polynomial time to sparse instances.
The proof of this result is based on the Packing Lemma of
\textcite{DellVanMelkebeek}.

Obtaining tight bounds for \matching{H} in general seems to be a challenging problem.
As Theorem~\ref{thm:cliquepacking} shows in the case of cliques, the lower bound of $O(k^{2-\epsilon})$ for \matching{H} implied by Theorem~\ref{thm:Hfactor} is not always tight.
On the other hand, the upper bound of~$O(k^{|V(H)|})$ is not always tight either:
A simple high-degree reduction rule shows that if~$H$ is a star of arbitrary size, then kernel with $O(k^2)$ edges are possible, which is tight by Theorem~\ref{thm:Hfactor}.
Furthermore, if $H=P_3$ is a path on 3 edges, then a surprisingly
nontrivial extremal argument gives us the following.
\begin{theorem}\label{th:3path}
	\matching{P_3} has kernels with $O(k^{2.5})$ edges.
\end{theorem}
We don't know currently if the exponent 2.5 is tight or the kernel size can be perhaps improved to $O(k^2)$.  Theorem~\ref{th:3path} raises the obvious question of how the exponent behaves for \matching{P_d} if $d$ increases (where $P_d$ is the path on $d$ edges): does it go to infinity or is there a universal constant bound on the exponent of the kernel size for every \matching{P_d} problem? We show that there is a ``uniformly polynomial'' kernel for~\matching{P_d} for any $d$:
\begin{theorem}\label{th:uniformpath}
	For \matching{P_d} (with $d$ being part of the input),
	\begin{enumerate}
		\item we can compute in time $n^{O(1)}$ a kernel of size $f(d)\cdot k^{O(1)}$, and
		\item we can compute in time $f(d)\cdot n^{O(1)}$ a kernel of size $k^{O(1)}$,
	\end{enumerate}
	where $f(d)$ is some computable function.
\end{theorem}
We cannot expect to improve Theorem~\ref{th:uniformpath} to kernels that can be computed in time~$n^{O(1)}$ and achieve size~$k^{O(1)}$ or even $(kd)^{O(1)}$, because setting $k=1$ would then yield a polynomial kernel for the $d$-Path problem, which is known to imply \hypofails. The significance of these subtle differences in the running
times/kernel sizes shows that one has to be very careful when talking about ``uniformly polynomial'' kernelization for a family of problems, as small changes in the definitions can matter a lot. In Section~\ref{sec:mult-kern}, we examine how \pp{Path Matching}, \pp{Clique Matching}, and \pp{Set Matching} behave under subtly different definitions. Uniform kernelization was studied also by \textcite{DBLP:journals/talg/GiannopoulouJLS17}, in the context of hitting forbidden minors.

The examples of cliques, stars, and paths show that the exact bound on
the kernel size of \matching{H} for a particular $H$ could be
very far from the $O(k^{|V(H)|})$ upper bound provided by the sunflower kernel
or the $O(k^{2-\epsilon})$ lower bound of Theorem~\ref{thm:Hfactor}.
It seems to be very challenging to obtain tight upper and lower bounds on the kernel sizes for all~$H$.
Our proofs of Theorem~\ref{th:3path} and Theorem~\ref{th:uniformpath} indicate what kind of
combinatorial problems we have to understand for a full solution. However, we want to point out that the algorithms in these proofs are very ``fragile'' in that they heavily rely on extremal properties of paths in graphs, and we were unable to extend them even in minor ways, such as to the problem of matching paths of length $d$ starting in the set $X$, or going from set $X$ to set $Y$.

After obtaining our results, we learnt that \textcite{HermelinWu}
independently achieved kernel lower bounds for
packing problems using the paradigm of Lemma~\ref{lemma:kernel-proxy}.
In particular, their bounds for \setmatching{d} and
\matching{K_d} are $O(k^{d-3-\epsilon})$ and $O(k^{d-4-\epsilon})$,
respectively.

\section{Techniques}

The \emph{\OR{} of a language $L$} is the language $\OR(L)$ that
consists of all tuples $(x_1,\dots,x_t)$ for which there is an
$i\in[t]$ with $x_i\in L$.
Instances $\overline{x}=(x_1,\dots,x_t)$ for $\OR(L)$ have two
natural parameters:
the length~$t$ of the tuple and the maximum bitlength $s=\max_i|x_i|$
of the individual instances for~$L$.
The following lemma captures the method that was used by
\textcite{DellVanMelkebeek} to prove conditional kernel lower bounds.
\begin{lemma}
	\label{lemma:kernel-proxy}
	Let $\Pi$ be a problem parameterized by~$k$
	and let~$L$ be an \NP-hard problem.
	Assume that there is a polynomial-time mapping reduction~$f$ from
	$\OR(L)$ to $\Pi$ and a number $d>0$ with the following property:
	given an instance $\overline x=(x_1,\dots,x_t)$ for $\OR(L)$ in
	which each $x_i$ has size at most $s$, the reduction produces an
	instance $f(\overline x)$ for $\Pi$ whose parameter $k$ is at most
	$t^{1/d+o(1)}\cdot\poly(s)$.

	Then $\Pi$ does not have kernels of size $O(k^{d-\epsilon})$ for any
	$\epsilon>0$ unless \hypofails.
\end{lemma}
\textcite{BDFH} formulated this method without the
dependency on~$t$.
This suffices to prove polynomial kernel lower bounds since~$d$ can be
chosen as an arbitrarily large constant.
\textcite{DellVanMelkebeek} adapted the proofs
of~\textcite{FortnowSanthanam,BDFH} to obtain the
formulation above, and generalized it to an oracle
communication setting.

We now informally explain a simple scheme for proving kernel lower
bounds of the form~$O(k^{d-\epsilon})$ for a parameterized problem $\Pi$.
Lemma~\ref{lemma:kernel-proxy} requires us to devise a reduction from
$\OR(L)$ (for some \NP-hard language $L$) to $\Pi$ whose output
instances have parameter~$k$ at most $t^{1/d}\cdot\poly(s)$. We
carefully select a problem $L$ whose definition is $d$-partite in
a certain sense, and we design the reduction from $\OR(L)$ to~$\Pi$
using the general scheme described.
Most problem parameters can be bounded from above by the number of
vertices; therefore, what we need to ensure is that the number of
vertices increases roughly by at most a factor of~$t^{1/d}$.

For simplicity of notation, we informally describe the case $d=2$
first. We assume that~$L$ is a bipartite problem, meaning that each
instance is defined on two sets $U$ and~$W$, and ``nothing interesting
is happening inside $U$ or inside $W$.''  We construct the instance of
$\Pi$ by taking $\sqrt{t}$ copies of~$U$ and~$\sqrt{t}$ copies
of~$W$. For each of the $t$ instances of $L$ appearing in the $\OR(L)$ instance, we select a copy $U'$ of $U$ and a copy $W'$ of $W$, and we embed the instance of $L$ into the union $U'\cup W'$. This way, we can embed all
$\sqrt{t}\cdot \sqrt{t}=t$ instances of $L$ such that each pair $(U',W')$ is selected for exactly one instance of $L$. The fact that $L$
is a bipartite problem helps ensuring that two instances of $L$
sharing the same copy of $U$ or the same copy of $W$ do not
interfere. A crucial part of the reduction is to ensure that every
solution of the constructed instance can use at most one copy of $U$
and at most one copy of $W$. If we can maintain this property (using
additional arguments or introducing gadgets), then it is usually easy
to show that the constructed instance has a solution if and only if at
least one of the $\sqrt{t}\cdot \sqrt{t}$ instances appearing in its
construction has a solution.

For $d>2$, the scheme is similar. We start with a $d$-partite problem
$L$ and make $t^{1/d}$ copies of each partition class. Then there are
$(t^{(1/d)})^d=t$ different ways of selecting one copy from each
class, and therefore we can compose together $t$ instances following
the same scheme.

As a specific example, let us consider $\Pi=\pp{Vertex Cover}$ in graphs, where we have $d=2$. The known kernel lower bound for this problem can be reproved more elegantly by choosing~$L$ in a not completely obvious way. In particular, we let $L$ be \pp{Multicolored Biclique}:
\begin{description}
	\item[Input:]
	      A bipartite graph $B$ on the vertex set $U\dotcup W$,
	      an integer~$k$,
	      and partitions $U=(U_1,\dots,U_k)$ and $W=(W_1,\dots,W_k)$.
	\item[Decide:]
	      Does $B$ contain a biclique $K_{k,k}$ that has a
	      vertex from each of the sets $U_1,\dots,U_k$ and $W_1,\dots,W_k$?
\end{description}
This is a problem on bipartite graphs and \NP-complete as we prove in
Appendix~\ref{appendix:multicolor-biclique}.
\begin{theorem}[\cite{DellVanMelkebeek}]
  \label{theorem:vertexcover}%
  If \hypo{}, then \pp{Vertex Cover} does not have kernels with $O(k^{2-\epsilon})$ edges.%
\end{theorem}
\begin{proof}
	We apply Lemma~\ref{lemma:kernel-proxy}
	with $L=\pp{Multicolored Biclique}$.
  To this end, we devise a reduction that is given an instance $(B_1,\dots,B_t)$ for $\OR(L)$ and outputs one instance of \pp{Vertex Cover} with at most $\sqrt{t} \cdot\poly(s)$ vertices, where~$s$ is the largest size among the~$B_1,\dots,B_t$.
  We can assume without loss of generality that every instance~$B_i$ has the same number $k$ of blocks in the partitions and every block in every instance~$B_i$
	has the same size $n$: we achieve this property by adding at most $\poly(s)$ isolated vertices to each instance, which will not affect the correctness or the asymptotic size bound.
	Similarly, we can assume that $\sqrt{t}$ is an integer. In the following, we refer to the $t$ instances of~$L$ instance as $B_{(i,j)}$ for $i,j\in\set{1,\dots,\sqrt{t}}$, and we let $U_{(i,j)}$ and $W_{(i,j)}$ be the two parts of the bipartite graph $B_{(i,j)}$.

  After this preprocessing, the reduction modifies each instance $B_{(i,j)}$ so that $U_{(i,j)}$ and $W_{(i,j)}$ become complete $k$-partite graphs; more precisely, if 	two vertices in $U_{(i,j)}$ or two vertices in $W_{(i,j)}$ are in different blocks (that is, they have different colors), then we make them adjacent.
  The new graph $B'_{(i,j)}$ has a~$2k$-clique if and only if there is a correctly partitioned $K_{k,k}$ in $B_{(i,j)}$.

	Next, the reduction constructs a graph $G$ by introducing $2\sqrt{t}$ sets $U^{1}$,
	$\dots$, $U^{\sqrt{t}}$ and $W^{1}$, $\dots$, $W^{\sqrt{t}}$ of $kn$
	vertices each. For every $i,j\in\set{1,\dots,\sqrt{k}}$, we copy the
	graph $B'_{(i,j)}$ to the vertex set $U^i\cup W^j$ by injectively mapping
	$U_{(i,j)}$ to $U^i$ and $W_{(i,j)}$ to $W^j$. Note that $U_{(i,j)}$ and $W_{(i,j)}$
	induces the same complete $k$-partite graph in $B'_{(i,j)}$ for
	every $i$ and $j$, thus this copying can be done in such a way that
	$G[U^i]$ receives the same set of edges when copying $B'_{(i,j)}$
	for any $j$ (and similarly for $G[W^j]$). Therefore, $G[U^i\cup
				W^j]$ is isomorphic to $B'_{(i,j)}$ for every $i,j$.

	We claim that $G$ has a $2k$-clique if and only if at least one
	$B'_{(i,j)}$ has a $2k$-clique (and therefore at least one $B_{(i,j)}$ has a
	correctly partitioned $K_{k,k}$).  The reverse direction is clear,
	as $B'_{(i,j)}$ is a subgraph of $G$ by construction and thus every clique in~$B'_{(i,j)}$ is also a clique in~$G$. For the
	forward direction, observe that $G$ has no edge between $U^{i}$ and~$U^{i'}$ if $i\ne i'$, and between $W^j$ and $W^{j'}$ if
	$j\neq j'$. Therefore, the $2k$-clique of $G$ is fully contained in $G[U^i\cup W^j]$ for some $i,j$. As $G[U^i\cup W^j]$ is isomorphic to~$B'_{(i,j)}$, this means that $B'_{(i,j)}$ also has a $2k$-clique.

  Let $N$ be the number of vertices in $G$ and note that $N=2\sqrt{t}\cdot kn\le t^{1/2}\cdot \poly(s)$ holds.
  The graph $G$ has a $2k$-clique if and only if its complement $\overline G$ has a vertex cover of size $N-2k$. So, if at the end our reduction outputs~$\overline G$, the algorithm we constructed is indeed a polynomial-time reduction from $\OR(L)$ to \pp{Vertex Cover} that achieves the size bound required for Lemma~\ref{lemma:kernel-proxy}.
\end{proof}
In Appendix~\ref{appendix:dVC}, we extend the elegant version of this proof to the vertex cover problem for $d$-uniform hypergraphs.

\section{Kernelization of the Set Matching Problem}\label{sec:packingproblems}
The \setmatching{d} problem is given a \mbox{$d$-uniform} hypergraph to
find at least~$k$ disjoint hyperedges.
For $d=2$, this is the maximum matching problem in graphs,
which is polynomial-time solvable.
The restriction of the problem to $d$-partite hypergraphs and $k=n/d$ is the
\perfectdimensionalsetmatching{d} and \NP-hard~\parencite{Karp72} for $d\geq 3$.

We use Lemma~\ref{lemma:kernel-proxy} to prove that the  kernel
size in Theorem~\ref{thm:setmatching-kernels} is asymptotically optimal under the hypothesis \hypo.
For the reduction, we use gadgets with few vertices that coordinate
the availability of groups of vertices.
For example, we may have two sets $U_1,U_2$ of $s$~vertices each and
our gadget makes sure that, in every perfect matching of the graph, one set is
fully covered by the gadget while all vertices of the other group have to
be covered by hyperedges external to the gadget (see Figure~\ref{fig:evencycle}).
Ultimately, we design gadgets that enable us to select exactly one of the
$m=t^{1/d}$ blocks in each of the $d$ color classes, which will in turn
activate one of the $t=m^d$ instances of the \OR-problem.
The precise formulation of the gadget is as follows.
\begin{lemma}[Selector gadget]\label{lem:setpacking-choice}%
	Let $d\geq 3$, $m\geq 1$, and $s\geq 1$ be integers.
	There is a $d$-partite $d$-uniform hypergraph~$S$ with $O(dms)$
	vertices and disjoint blocks $U_1,\dots,U_m \subseteq V(S)$, each of
	size $s$ and in the same color class of $S$, such that the following
	conditions hold.
	\begin{enumerate}[label=(\roman*)]
		\item (Completeness)
		      $S-U_i$ has a unique perfect matching for all $i\in[m]$.
		\item (Soundness)
		      For all sets $B\subseteq U_1\cup\dots\cup U_m$ so that $B\neq
			      U_i$ for any $i\in[m]$, the graph $S-B$ does not have a perfect
		      matching.
	\end{enumerate}
	Moreover, the graph can be computed in polynomial time.
\end{lemma}
The fact that the selector gadget is $d$-partite gives us the result for
\perfectdimensionalsetmatching{d} rather than just for \setmatching{d}, and it
is also useful in the reduction to \matching{K_d}.
We defer the proof of Lemma~\ref{lem:setpacking-choice} to
\S\ref{sec:selector-proof}, and use it now to prove the following reduction.
\begin{lemma}\label{lemma:setpacking-reduction}
	Let $d\geq 3$ be an integer.
	There is a polynomial-time mapping reduction from $\OR($\perfectdimensionalsetmatching{d}$)$ to
	\perfectdimensionalsetmatching{d} that maps
	$t$-tuples of instances of bitlength $s$ each to instances on
	$t^{1/d}\cdot\poly(s)$ vertices.
\end{lemma}
\begin{proof}
	Let $G_1,\dots,G_t$ be instances of
	\perfectdimensionalsetmatching{d}, that is, $d$-partite $d$-uniform
	hypergraphs of size~$s$ each.
	We can assume w.l.o.g.\ that each color class of each $G_i$ contains
	exactly $n/d$ vertices where $n\leq s$.
	The goal is to find out whether some $G_i$ contains a perfect
	matching.
	We reduce this question to a single instance $G$ on few vertices.

	The vertex set of $G$ consists of $d\cdot t^{1/d}$ groups of $n/d$
	vertices each, i.e., $V(G)=\bigcup_{a,b} V_{a,b}$ for $a\in[d]$ and
	$b\in[t^{1/d}]$.
	Then we can write the input graphs as $G_{b}$ using an
	index vector $b=(b_1,\dots,b_d)\in[t^{1/d}]^d$.
	For each graph $G_{b}$ we add edges to $G$ in the
	following way:
	We identify the vertex set of $G_{b}$ with
	$V_{1,b_1}\dotcup\dots\dotcup V_{d,b_d}$,
	and we let $G$ contain all the edges of~$G_{b}$.
	Since each $G_b$ is $d$-partite, the same is true for $G$ at this
	stage of the construction.
	Now we modify $G$ such that each perfect matching of $G$ only ever
	uses edges originating from at most one graph $G_b$.
	For this it suffices to add a gadget for every $a\in [d]$ that blocks all
	but exactly one group $V_{a,b}$ in every perfect matching.
	For each $a\in [d]$, we add an independent copy~$S_a$ of the
	gadget~$S$ from Lemma~\ref{lem:setpacking-choice} to~$G$, where we
	identify the sets $U_1=V_{a,1}$ up to $U_m=V_{a,m}$.
	Clearly, $|V(G)|\leq O(st^{1/d})$, and the graph $G$ is $d$-partite
	as this is true for the input graphs and the gadgets.

	Now we verify the correctness of the reduction.
	If some $G_b$ has a perfect matching, then the completeness property
	of $S_a$ ensures that $S_a - V_{a,b_a}$ has a unique perfect
	matching for all $a\in [d]$.
	Together with the perfect matching of $G_b$, this gives a perfect
	matching of~$G$.

	For the soundness, assume $M$ is a perfect matching of $G$.
	Consider now the copy of $S_a$ introduced into the construction. The vertex set of this gadget is $\bigcup_{i=1}^m V_{a,i}$ and a set $Z_a$ of vertices private to this gadget. Let $M_a\subseteq M$ be the set of those edges that contain at least one vertex of $Z_a$. Then $M_a$ covers every vertex of $Z_a$ and a subset of $\bigcup_{i=1}^m V_{a,i}$. Now Lemma~\ref{lem:setpacking-choice}(ii) implies that there is a $b_a$ such that $M_a$ covers every vertex of  $\bigcup_{i=1}^m V_{a,i}$ except $V_{a,b_a}$. It follows that $M^*=M\setminus \bigcup_{a\in[d]}M_a$ covers every $V_{a,b_a}$.
	Since $V_{a,b_a}$ is an independent set in $S_a$, we hav e that $M^*$ uses only
	edges of $G_b$ to cover the $V_{a,b_a}$.
	In particular, $G_b$ has a perfect matching.
\end{proof}
Theorem~\ref{thm:setpacking}, our kernel lower bound for
\setmatching{d}, now follows immediately by combining the above with
Lemma~\ref{lemma:kernel-proxy}.

\subsection{Construction of the Selector Gadget}
\label{sec:selector-proof}
\begin{proof}[of Lemma~\ref{lem:setpacking-choice}]
	We first implement the selector gadget~$S_m$ for $m=2$ and then use this as a
	building block for $m>2$.
	\begin{figure}[tpb]
		\begin{center}
			\begin{tikzpicture}
				\tikzstyle{node}=[fill,circle,inner sep=1.5pt]

				\def\numpts{6}
				\def\radius{1}

				\foreach\x in {1,...,\numpts}
					{
						\ifthenelse{\isodd{\x}}{
							\path[fill=orange!70!black]
							(\x*360/\numpts+90: \radius)
							to[bend right=80] (\x*360/\numpts+360/\numpts+90:
							\radius+.1)
							to (\x*360/\numpts+360/\numpts+90: \radius-.1)
							to[bend right=20] (\x*360/\numpts+90: \radius);
						}{
							\path[fill=orange!70!white]
							(\x*360/\numpts+90: \radius-.1)
							to (\x*360/\numpts+90: \radius+.1)
							to[bend right=80] (\x*360/\numpts+360/\numpts+90: \radius)
							to[bend right=20] (\x*360/\numpts+90: \radius-.1);
						}
					}

				\foreach\x in {1,...,\numpts}
					{
						\ifthenelse{\isodd{\x}}{
							\node[node,black!65!white] at (\x*360/\numpts-360/\numpts+90:
							\radius-.1) { } ;
							\node[node,black!65!white] at (\x*360/\numpts-360/\numpts+90:
							\radius+.1) { } ;
						}{
							\node[node,black!65!white] at (\x*360/\numpts-360/\numpts+90: \radius) { } ;
						}

						\node[node] at (\x*360/\numpts-360/\numpts/2+90: \radius) {};
						\node[anchor=center] at (\x*360/\numpts-360/\numpts/2+90:
						\radius+.4) {$\x$};
					}%
			\end{tikzpicture}
			\hspace{.3cm}\vrule\hspace{.3cm}
			\begin{tikzpicture}[label distance=.2cm]
				\tikzstyle{node}=[fill,circle,inner sep=1.5pt]

				\draw[rounded corners=2]
				(-2.5,-.2) rectangle (-.5,.2);
				\draw[rounded corners=2]
				(.5,-.2) rectangle (+2.5,.2);
				\draw
				(-.5,0)--(.5,0) node[midway,fill=white,inner sep=0]
				{$\oplus$};
				\foreach\x in {1,3,5}
				{
				\node[node,label=above:{$\x$}] at (-2.25+.25*\x, 0) { } ;
				}
				\foreach\x in {2,4,6}
				{
				\node[node,label=above:{$\x$}] at (.5+.25*\x, 0) { } ;
				}
				\begin{scope}[yshift=-2cm,xshift=-3.5cm]
					\draw (.4,1) -- (6.6,1);
					\draw[rounded corners=2]
					(0,-.2) rectangle (1.5,.2) node[midway] {$U_1$};
					\draw[rounded corners=2]
					(2,-.2) rectangle (3.5,.2) node[midway] {$U_2$};
					\draw[rounded corners=2]
					(5.5,-.2) rectangle (7,.2) node[midway] {$U_m$};
					\draw (1.5,0)--(2,0) node[midway,fill=white,inner sep=0]
					{$\oplus$};
					\draw (3.5,0)--(4,0) node[midway,fill=white,inner sep=0]
					{$\oplus$};
					\node[inner sep=0] at (4.5,0) {$\dots$};
					\draw (5,0)--(5.5,0) node[midway,fill=white,inner sep=0]
					{$\oplus$};
				\end{scope}
			\end{tikzpicture}
		\end{center}
		\caption{\label{fig:evencycle}%
    \emph{Left:}
			A \emph{switch gadget} with $d=4$, $s=3$, $U_1=\{1,3,5\}$, and
			$U_2=\{2,4,6\}$.
			Black bullets represent external vertices, and gray bullets
			represent vertices private to the gadget.
			The gadget has two states, that is, two ways in which it can be part of a
			perfect matching: Either all hyperedges with the darker shading or all
			hyperedges with the lighter shading are contained in the perfect matching.
			In the former case, $U_1$ must be covered by hyperedges of the
			outside graph, and in the latter case $U_2$ must be.
			\emph{Top Right:}
			Pictorial abbreviation of the graph on the left.
			By Lemma~\ref{lem:setpacking-choice}, any perfect matching
			blocks exactly the vertices in one of the halves using edges of
			the gadget.
			\emph{Bottom Right:}
			A \emph{selector gadget} is a composition of $m$ switch gadgets.
			By Lemma~\ref{lem:setpacking-choice}, such a composition has
			exactly $m$ possible states: exactly one block of external
			vertices~$U_i$ for $i\in[m]$ is left free and all other blocks~$U_j$ for
			$j\neq i$ are fully covered by the hyperedges of the switch gadgets.}
	\end{figure}
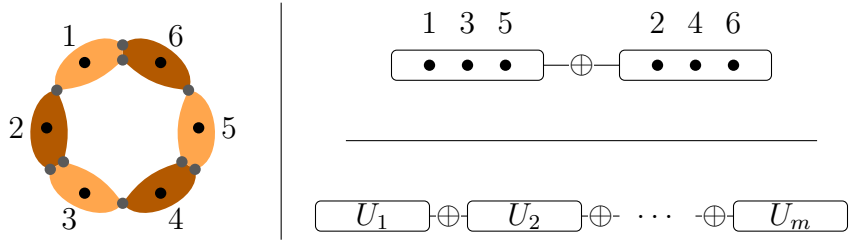%
	The \emph{switch} gadget $S=S_2$ is a $d$-uniform cycle $e_1,\dots,e_{2s}$ as
	depicted on the left in Figure~\ref{fig:evencycle}.
	Formally, we impose the following structure (where indices are
	understood modulo $2s$):
	\begin{itemize}
		\item
		      $|e_{i}\cap e_{i+1}| = 1$ if $i$ is odd.
		\item
		      $|e_{i}\cap e_{i+1}| = d-2$ if $i$ is even.
		\item
		      $|e_{i}\cap e_{j}| = 0$ if $j \not\in \{i-1,i,i+1\}$.
	\end{itemize}
	The graph contains $O(ds)$ vertices.
	We set the two sets of size $s$ that we want to switch between as
	$U_1=\bigcup_{\text{odd } i} e_{i} \setminus \bigcup_{\text{even }i} e_{i}$
	and $U_2=\bigcup_{\text{even }i} e_{i} \setminus \bigcup_{\text{odd
			}i} e_{i}$.
	To show the required properties, let us first note that the hypergraph graph is
	$d$-partite:
	The vertices of $U_1\cup U_2$ form color class $1$, the vertices appearing in
	the intersections~$e_i\cap e_{i+1}$ for odd $i$ form color class $2$, and the
	remaining vertices form color classes $3$ through $d$.
	For the soundness, note that removing $U_1$ from the graph also
	removes the edges~$e_i$ with odd indices, and the edges that remain
	form a perfect matching of $S-U_1$.
	The case of removing $U_2$ is analogous.
	For the completeness, note that removing any set $B\subseteq U_1\cup
		U_2$ other than $U_1$ or $U_2$ from $S$ leaves a graph that does not
	have a perfect matching.

	Now let $m>2$ and let $U_1,\dots,U_m$ be disjoint sets of size~$s$
	each.
	We construct~$S=S_m$ as on the bottom right of
	Figure~\ref{fig:evencycle} by adding independent copies of the graph
	$S_2$ between~$U_1$ and $U_2$, between $U_2$ and $U_3$, and in general
	between $U_i$ and $U_{i+1}$ for all $i\in[m-1]$.
	The graph is trivially $d$-partite since this was true for $S_2$ and
	the sets $U_1\cup\dots\cup U_m$ are all in the same color class, say
	$1$; the number of vertices of $S$ is $O(dms)$.
	For the completeness, consider the graph $S-U_i$ for some $i\in[m]$.
	By the completeness of $S_2$, the two instances of $S_2$ that are
	incident to $U_i$ can be completed to a perfect matching that covers
	$U_{i-1}$ and $U_{i+1}$.
	We can continue applying the completeness property of $S_2$
	inductively to find the unique perfect matching of $S-U_i$.

  The soundness property of $S_2$ implies that
  ${B\cap (U_i \cup U_{i+1})}$ must be equal to $U_i$ or~$U_{i+1}$
  for all $i$ or else there cannot be a perfect matching.
	Moreover, assume that $B$ contains more than one set $U_i$.
	By the soundness of $S_2$, these sets cannot be adjacent or else there
	is no perfect matching.
	Let us for simplicity assume that $B$ contains exactly two non-adjacent sets,
	say $B=\{U_1,U_i\}$, and we assume for contradiction that $S-B$ does not contain
	a perfect matching.
	Then the completeness property of~$S_2$ implies that there is exactly one
	perfect matching between $U_{i-1}$ and $U_{i}$ that does not use vertices from
	$U_i$ but covers all vertices of $U_{i-1}$.
	Thus $S-(U_1\cup U_{i-1})$ also contains a perfect matching.
	By induction, we can see that $S-(U_1\cup U_2)$ has a perfect matching, which we
	already ruled out.
	Overall, we get $B=U_i$ for some $i$.
\end{proof}

\section{Kernel Lower Bounds for Graph Matching Problems}
For a graph~$H$, the \matching{H} problem is given a graph $G$ to find a
maximal number of vertex-disjoint copies of~$H$.
This problem is \NP-complete whenever~$H$ contains a connected
component with more than two vertices~\parencite{KirkpatrickHell} and is in~$\P$ otherwise.

\subsection{Clique Packing}
We prove Theorem~\ref{thm:cliquepacking}, that \matching{K_d} for
$d\geq 3$ does not have kernels with at most~$O(k^{d-1-\epsilon})$ edges unless
\hypofails{}. For this, we devise a parameter-preserving reduction from the problem
of finding a perfect matching in a $(d-1)$-partite $(d-1)$-uniform
hypergraph.
\begin{lemma}
	\label{lem:cliquepacking}
	For all $d\geq 4$, there is a polynomial-time mapping reduction from
	\perfectdimensionalsetmatching{(d-1)} to \matching{K_d} that does
	not change the parameter~$k$.
\end{lemma}
\begin{proof}
	Let $G$ be a $(d-1)$-partite $(d-1)$-uniform hypergraph on~$n$
	vertices.
	For each edge~$e$ of~$G$, we add a new vertex~$v_e$ and transform
	$e\cup\{v_e\}$ into a $d$-clique in~$G'$.
	We claim that~$G$ has a matching of size $k\coloneqq n/(d-1)$ if and only
	if~$G'$ has a \mbox{$K_d$-matching} of size~$k$.
	The completeness is clear since any given matching of~$G$ can be
	turned into a \mbox{$K_d$-matching} of~$G'$ by taking the respective
	$d$-clique for every $(d-1)$-hyperedge.
	For the soundness, let $G'$ contain a \mbox{$K_d$-matching} of
	size~$k$.
	Note that any $d$-clique of~$G'$ uses exactly one vertex~$v_e$ since
	the underlying graph of~$G$ does not contain any $d$-cliques and
	since no two $v_e$'s are adjacent.
	Thus every $d$-clique of~$G'$ is of the form $e\cup\{v_e\}$, which
	gives rise to a matching of~$G$ of size~$k$.
\end{proof}
This combined with Lemma~\ref{lemma:kernel-proxy} and
Lemma~\ref{lemma:setpacking-reduction} implies
Theorem~\ref{thm:cliquepacking} for $d\geq 4$.
The case $d=3$ follows from Theorem~\ref{thm:Hfactor}, which we
establish independently in the following.

\subsection{General Graph Matching Problems}

We prove Theorem~\ref{thm:Hfactor}, that \factor{H} does not
have kernels of size $O(k^{2-\epsilon})$ unless \hypofails, whenever
$H$ is a connected graph with at least three vertices.

We use the coordination gadget of Lemma~\ref{lem:setpacking-choice} in
a reduction from a suitable \OR-problem to \mbox{\matching{H}}.
To do so, we translate the coordination gadget for
\perfectdimensionalsetmatching{d} to \mbox{\factor{H}},
which we achieve by replacing hyperedges with the following
hyperedge-gadgets of~\textcite{KirkpatrickHell}.

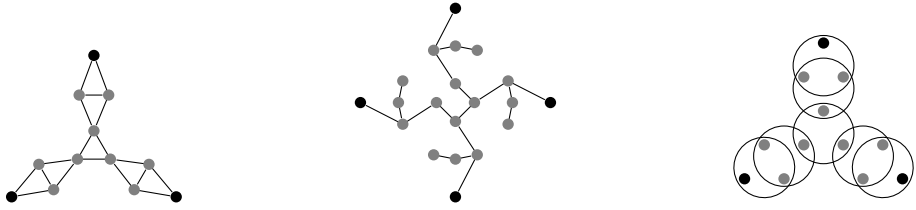
\begin{figure}[tpb]
	\begin{center}
		\begin{tikzpicture}
			\tikzstyle{node}=[fill,circle,inner sep=1.5pt]

			\def\numpts{3}
			\def\radius{.25}

			\foreach\x in {1,...,\numpts}
				{
					\draw
					(\x*360/\numpts-360/\numpts+90: \radius)
					--
					(\x*360/\numpts+90: \radius)
					;
				}
			\foreach\x in {1,...,\numpts}
				{
					\node[node,gray] (inside)
					at (\x*360/\numpts-360/\numpts+90: \radius) { } ;
					\node[node] (outside)
					at (\x*360/\numpts-360/\numpts+90: \radius+1) { } ;
					\draw
					(inside)
					--
					(\x*360/\numpts-360/\numpts+45/\numpts+90: \radius + .5)
					node[node,gray] (help1) { }
					--
					(outside)
					;
					\draw
					(inside)
					--
					(\x*360/\numpts-360/\numpts-45/\numpts+90: \radius + .5)
					node[node,gray] (help2) { }
					--
					(outside)
					;
					\draw (help1) -- (help2);
				}
		\end{tikzpicture}
		\hspace{2cm}
		\begin{tikzpicture}
			\tikzstyle{node}=[fill,circle,inner sep=1.5pt]

			\def\numpts{4}
			\def\radius{.25}

			\foreach\x in {2,...,\numpts}
				{
					\draw
					(\x*360/\numpts-360/\numpts+90: \radius)
					--
					(\x*360/\numpts+90: \radius)
					;
				}
			\foreach\x in {1,...,\numpts}
				{
					\node[node,gray] (inside)
					at (\x*360/\numpts-360/\numpts+90: \radius) { } ;
					\node[node] (outside)
					at (\x*360/\numpts-360/\numpts+90: \radius+1) { } ;
					\draw
					(inside)
					--
					(\x*360/\numpts-360/\numpts+90/\numpts+90: \radius + .5)
					node[node,gray] (help1) { }
					--
					(outside)
					;
					\draw
					(\x*360/\numpts-360/\numpts-90/\numpts+90: \radius + .5)
					node[node,gray] { }
					-- (\x*360/\numpts-360/\numpts+90: \radius + .5)
					node[node,gray] { }
					-- (help1);
				}
		\end{tikzpicture}
		\hspace{2cm}
		\begin{tikzpicture}
			\tikzstyle{node}=[fill,circle,inner sep=1.5pt]

			\def\numpts{3}
			\def\radius{.3}

			\foreach\outerH in {1,...,\numpts}
				{
					\draw (\outerH*360/\numpts-360/\numpts+90:.6) circle[radius=.4];

					\begin{scope}[rotate=\outerH*360/\numpts-360/\numpts,yshift=.9cm]
						\draw (0,0) circle[radius=.4];
						\foreach\x in {1}
							{
								\node[node] (inside)
								at (\x*360/\numpts-360/\numpts+90: \radius) { } ;
							}
						\foreach\x in {2,...,\numpts}
							{
								\node[node,gray] (inside)
								at (\x*360/\numpts-360/\numpts+90: \radius) { } ;
							}
					\end{scope}
				}

			\draw (0,0) circle[radius=.4];
			\foreach\x in {1,...,\numpts}
				{
					\node[node,gray] (inside)
					at (\x*360/\numpts-360/\numpts+90: \radius) { } ;
				}
		\end{tikzpicture}
	\end{center}
	\caption{\label{fig:Hpacking-hyperedge}%
		Hyperedge gadgets for different $H$-matching problems.
		The outermost, black vertices are the vertices of the simulated
		hyperedge and the gray vertices are not supposed to be adjacent to
		any other vertex of the graph.
		\emph{Left:}
		Triangle matching.
		\emph{Middle:}
		$3$-Path matching.
		\emph{Right:}
		The general case; each circle represents a copy of~$H$.}
\end{figure}

\begin{lemma}\label{lem:Hpacking-hyperedge}%
	Let $H$ be a connected graph on $d\geq 3$ vertices.
	There is a graph $e=e(v_1,\dots,v_d)$ that contains
	$\{v_1,\dots,v_d\}$ as an independent set such that, for all
	$S\subseteq \{v_1,\dots,v_d\}$,
	the graph $e-S$ has an $H$-factor
	if and only if
	$|S|=0$ or $|S|=d$.
\end{lemma}
\begin{proof}
	Let $v$ be a vertex of $H$%
	.
	We construct~$e$ as in Figure~\ref{fig:Hpacking-hyperedge}.
	We start with one central copy of $H$.
	For each vertex $u\in[d]=V(H)$,
	we create a new copy $H_u$ of $H$ and denote its copy of $v$ by
	$v_u$.
	Finally, we add an edge between $u\in H$ and $w\in (H_u-v_u)$ if
	$v_u w$ is an edge of $H_u$.

	For the claim, assume that $0<|S|<d$.
	Then $|V(e-S)|$ is not an integer multiple of $d=|V(H)|$ and there
	can be no $H$-factor in $e-S$.
	For the other direction, assume that~${|S|=0}$.
	Then the subgraphs $H_u$ for $u\in[d]$ and $H$ are $d+1$ pairwise
	disjoint copies of $H$ in $e$ and form an $H$-factor of $e$.
  In the case $|S|=d$, we observe that the $d$ subgraphs~${(H_u-v_u)\cup\{u\}}$
  form an $H$-factor of	$e-S=e-\{v_1,\dots,v_d\}$.
\end{proof}

For the proof of the \matching{H} kernel lower bounds, we use the
Packing Lemma. Recall that a set $E$ of edges in a $d$-uniform hypergraph is a clique on $p$ vertices if there is a set $S$ of $p$ vertices such that $E$ contains exactly the $\binom{p}{d}$ size-$d$ subsets of $S$.
\begin{lemma}[\cite{DellVanMelkebeek}]\label{lem:packing}%
	For all integers $p \geq d \geq 2$ and $t>0$
	there is a $p$-partite $d$-uniform hypergraph $P$ on
	$O\big(p\cdot\max(p,t^{1/d+o(1)})\big)$ vertices such that
	\begin{enumerate}[label=(\roman*)]
		\item the hyperedges of $P$ are partitioned into~$t$ cliques
		      $K_1,\dots,K_t$ on~$p$ vertices each, and
		\item $P$ contains no cliques on $p$ vertices other than the~$K_i$'s.
	\end{enumerate}
	Furthermore, for any fixed $d$, the hypergraph $P$ and the $K_i$'s can be
	constructed in time polynomial in~$p$ and~$t$.
\end{lemma}
The chromatic number $\chi(H)$ is the minimum number of colors
required in a proper vertex-coloring of~$H$.
The proof of \textcite{KirkpatrickHell} shows that \factor{H} is
\NP-complete even in case we are looking for an $H$-factor in
$\chi(H)$-partite graphs.
We are going to make use of that in the following reduction.
\begin{lemma}\label{lem:Hfactor}%
	There is a polynomial-time mapping reduction from $\OR(\factor{H})$ to
	\factor{H} that maps $t$-tuples of instances
	of size $s$ each to instances that have at most
	$\sqrt{t}^{1+o(1)}\cdot \poly(s)$ vertices.
\end{lemma}
\begin{proof}
	Let $p=\chi(H)$ be the chromatic number of $H$.
	For an instance $G_1,\dots,G_t$ of $\OR(\factor{H})$,
	we can assume w.l.o.g.\ that the $G_i$ are $p$-partite graphs
	with~$n$ vertices in each part.
	We construct a graph~$G$ that has an $H$-factor if and only if some
	$G_{i}$ has an $H$-factor.
	For this, we invoke the Packing Lemma, Lemma~\ref{lem:packing}, with
	$d=2$, and we obtain a $p$-partite graph~$P$ that
	contains~$t$ cliques $K_1,\dots,K_t$ on~$p$ vertices each.
	We identify the vertex set of $G_i$ with $V(K_i)\times [n]$
	injectively in such a way that vertices in the same color class
	have the same first coordinate.
	We define an intermediate $p$-partite graph $G'$ on the vertex set
	$V(P)\times [n]$ as $G'=G_1\cup\dots\cup G_t$.
	To obtain $G$ from $G'$, we add $p$ coordination gadgets of
	Lemma~\ref{lem:setpacking-choice} with $m=\sqrt t^{1+o(1)}$ and $d=p$.
	For each color class $C\subset V(G')$, we add a coordination gadget
	where the $U_i\subset C$ are those vertices that project to the same
	vertex in $P$.
	Finally, we replace each $p$-hyperedge by the gadget in
	Lemma~\ref{lem:Hpacking-hyperedge}, which finishes the construction
	of $G$.

	For the completeness of the reduction, assume $G_i$ has an
	$H$-factor~$M$.
	To construct an $H$-factor of $G$, we start by using $M$ to cover
	the vertices $V(G_i)$ in $G$.
	The completeness of the coordination gadgets guarantees that we find
	a perfect matching in the $d$-uniform hypergraph $G'-V(G_i)$ that
	uses only hyperedges of the coordination gadgets.
	By Lemma~\ref{lem:Hpacking-hyperedge}, this gives rise to an
	$H$-factor of $G$.

	For the soundness, assume we have an $H$-factor $M$ of $G$.
	Lemma~\ref{lem:Hpacking-hyperedge} guarantees that the edge gadgets
	can be seen as $p$-hyperedges in the intermediate graph~$G'$.
	Soundness of the coordination gadgets guarantees that $M$ leaves
	exactly one group free per part.
	Now let $H'$ be a copy of $H$ that is contained in $G$ but not in
	any of the gadgets.
	Since~$H'$ has chromatic number~$p$, $H'$ intersects all $p$ parts
	and has an edge between any two distinct parts.
	By construction of $G$, this implies that the projection of $H$ onto
	$P$ is a clique.
	By the packing lemma, this clique is one of the $K_i$'s.
	Therefore, each $H'$ of the $H$-factor~$M$ that is not in one of the
	gadgets is contained in~$G_i$, which implies that~$G_i$ has an
	$H$-factor.

	The claim follows since $G$ is a graph on
	$\sqrt{t}^{1+o(1)}\poly(s)$ vertices that has an $H$-factor
	if and only if some $G_i$ has an $H$-factor.
\end{proof}
Now Lemma~\ref{lemma:kernel-proxy} immediately implies
Theorem~\ref{thm:Hfactor}, our kernel lower bounds for
\factor{H}.

\section{Multiparameter Kernelization}
\label{sec:mult-kern}

The kernelization complexities of \pp{$d$-Set Matching} and \pp{$d$-Clique
	Matching} appear to be significantly different from that of
\pp{$P_d$-Matching}.
To make this observation formal, we consider multiparameter problems:
Let $\Pi$ be a problem that has a \emph{primary} parameter~$s$ and a
\emph{secondary} parameter~$d$.
Here, $s$ should be thought of as the solution size -- for packing problems,
$s=kd$ is the number of vertices in the solution.

We want to classify the kernelization complexity of $\Pi$ further, so we assume
that $\Pi$ is fixed-parameter tractable with respect to the solution size~$s$,
that is, it has an algorithm running in time~$f(s)\poly(n)$.
Here, $\poly(n)$ means $C \cdot n^C$ for some fixed constant~$C$, which is
independent from~$d$.
We classify the running time of a kernelization algorithm~$A$ for~$\Pi$ into the
following three classes:
\begin{enumerate}[label=(\Alph*)]
	\item
	      $A$ runs in time $\poly(n)$.
	\item
	      $A$ runs in time $f(d) \poly(n)$ for some computable non-decreasing
	      function~$f$.
	\item
	      $A$ runs in time $O(n^{f(d)})$ for some computable non-decreasing
	      function~$f$.
\end{enumerate}
Moreover, we distinguish the following asymptotic kernel sizes that~$A$ may
produce:
\begin{enumerate}[label=(\arabic*)]
	\item $\poly(s)$
	\item $g(d) \poly(s)$ for some computable non-decreasing function~$g$
	\item $O(s^{g(d)})$ for some computable non-decreasing function~$g$
\end{enumerate}
A priori, this gives us nine different possible kernelization complexities for a
multiparameter problem, where A1 is the best type and C3 is the worst.

For example, \pp{$d$-Set Matching} has kernels of size $O(d! k^d)$ that can be
computed in time $\poly(n)$, so it is of type~A3.
The running time is so fast because we can find sunflowers in $d$-uniform
hypergraphs in linear time.
Moreover, our lower bound from Theorem~\ref{thm:setpacking} shows that the problem is
unlikely to have type C2 or better.
This completely classifies the multiparameter kernelization complexity of
\pp{$d$-Set Matching} (see Figure~\ref{fig:results}).

As another example, \pp{$d$-Clique Matching} has kernels of size $O(d! k^d)$
that can be computed in time $O(n^{O(d)})$, which is a kernel of type $C3$.
For this problem, we cannot find sunflowers quickly since finding even one
$d$-clique is hard:
Assume there was a B3-kernelization for \pp{$d$-Clique Matching}.
Then this algorithm works in particular for $k=1$, in which case it runs in time
$f(d)\poly(n)$ and produces and output of size $O(s^{g(d)}) = O(d^{g(d)})$.
On this kernel, we can run any brute force algorithm for finding a clique;
overall, this gives an \FPT-algorithm for the $d$-Clique problem, which implies
$\FPT=\W[1]$.

In Section~\ref{sec:packing-paths-length}, we explore the multiparameter kernelization complexity of the \pp{$P_d$-Matching} problem.
We will find it useful to apply the following standard trick, which can turn any B2-kernelization into one of
type A2 and one of type B1.

\begin{lemma}\label{lem:B2toA2B1}
	Any kernelization algorithm of type B2 gives rise to kernelization algorithms
	of type A2 and B1.
\end{lemma}
In particular, a problem has kernels of type~A2 if and only if it has
kernels of type~B1.
\begin{proof}
	Let $A$ be a B2-type kernelization algorithm.
	We construct a B1-type kernelization~$B$ as follows:
	If $g(d) < k$, then we run $A$, which outputs an instance of size $\poly(k)$.
	Otherwise, $g(d) \ge k$.
	Since $\Pi$ is FPT with respect to $k$, we can solve it in time $h(k) \poly(n)
		\le h(g(d)) \poly(n)$ which reduces the kernel size to a constant.
	Since $h$ and $g$ are computable, $B$ is of type B1.
	We construct an A2-type kernelization~$C$ as follows:
	If $f(d) < n$, then we run $A$, which takes time $\poly(n)$.
	Otherwise, $f(d) \ge n$, which means that the instance is already sufficiently
	kernelized.
\end{proof}

\definecolor{nocolor}{rgb}{0.4,0,0}
\definecolor{yescolor}{rgb}{0,0.4,0}
\tikzstyle{no}=[line width=0.7mm,nocolor,dashed]
\tikzstyle{yes}=[line width=0.7mm,yescolor]
\def\setuptimesize{%
\node (TimeA) at (0,1.2)   {A};
\node (TimeB) at (0,.6)    {B};
\node (TimeC) at (0,0)     {C};
\node (Size1) at (1.5,1.2) {1};
\node (Size2) at (1.5,.6)  {2};
\node (Size3) at (1.5,0)   {3};
\node (Time)  at (0.1,1.7) {Time};
\node (Size)  at (1.4,1.7) {Size};
}%
\begin{figure}%
  \centering%
  \begin{tikzpicture}[baseline=(Label.base)]
    \setuptimesize
    \draw[no]  (TimeA) -- (Size1);
    \draw[yes] (TimeA) -- (Size2);
    \draw[yes] (TimeB) -- (Size1);
    \node (Label) at (0.5,-0.7) {\pp{Path Matching}};
    \node at (-0.25,1.2) [anchor=east] {\footnotesize$\poly(n)$};
    \node at (-0.25,0.6) [anchor=east] {\footnotesize$f(d)\poly(n)$};
    \node at (-0.25,0.0) [anchor=east] {\footnotesize$O(n^{f(d)})$};
	\end{tikzpicture}%
	\hspace{0.5cm}
	\begin{tikzpicture}[baseline=(Label.base)]
    \setuptimesize
		\draw[no]  (TimeC) -- (Size2);
		\draw[yes] (TimeA) -- (Size3);
    \node (Label) at (0.5,-0.7) {\pp{Set Matching}};
	\end{tikzpicture}%
	\hspace{0.5cm}
	\begin{tikzpicture}[baseline=(Label.base)]
    \setuptimesize
		\draw[no]  (TimeC) -- (Size2);
		\draw[no]  (TimeB) -- (Size3);
		\draw[yes] (TimeC) -- (Size3);
		\node (Label) at (0.5,-0.7) {\pp{Clique Matching}};
    \node at (1.75,1.2) [anchor=west] {\footnotesize$\poly(s)$};
    \node at (1.75,0.6) [anchor=west] {\footnotesize$g(d)\poly(s)$};
    \node at (1.75,0.0) [anchor=west] {\footnotesize$O(s^{g(d)})$};
  \end{tikzpicture}%
	\caption{\label{fig:results}%
		Multiparameter kernelization complexities of \pp{Path}, \pp{Set}, and \pp{Clique Matching}.
		The \emph{solid lines} indicate the existence of such a kernelization,
		whereas the \emph{dashed lines} indicate that such a kernel can only exist
		if $\hypofails$ or $\FPT=\W[1]$.
		The status of all combinations not drawn is implied by the ones that are.}
\end{figure}

\section{Kernels for Graph Packing Problems}

The sunflower kernelization in Theorem~\ref{thm:setmatching-kernels}
immediately transfers to \matching{H} for any fixed graph~$H$ and
yields kernels with $O(k^{d})$ edges.  For every graph $H$,
\textcite{moser_packing} shows that \matching{H} has kernels
with $O(k^{d-1})$ {\em vertices} where~$d=|V(H)|$, but this gives only
the weaker bound $O(k^{2d-2})$ on the number of edges.
Here we show that for some specific $H$, we can obtain kernels that
are better than the $O(k^{d})$ bound implied by
Theorem~\ref{thm:setmatching-kernels}.
As a very simple example, we show this first for
\matching{K_{1,d}}, the problem of packing vertex-disjoint stars
with $d$ leaves.
\begin{proposition}\label{prop:star}
	\matching{K_{1,d}} has kernels with $O(k^2)$ edges.
\end{proposition}
\begin{proof}
	Let $(G,k)$ be an instance of \matching{K_{1,d}}
	If $G$ has a vertex~$v$ of degree at least $dk+1$, let $e$ be an
	edge incident to~$v$.
	We claim that we can safely remove~$e$.
	If~$G-e$ has a $K_{1,d}$-matching of size $k$, then this also holds
	for $G$.
	For the other direction, let~$M$ be a \mbox{$K_{1,d}$-matching} of
	size $k$ in $G$.
	If $M$ does not contain $e$, it is also a matching of~$G-e$.
	Otherwise $M$ contains $e$.
	Let $M'$ be obtained from $M$ by removing the star that contains~$e$.
	Now $v$ is not contained in $M'$.
	Since $M'$ covers at most $d(k-1)$ vertices, at least $d+1$
	neighbors of $v$ are not contained in $M'$.
	Even if we remove $e$, we can therefore augment $M'$ with a
	vertex-disjoint star that is centered at $v$ and has $d$ leaves.
	This yields a star matching of size $k$ in $G-e$.

	For the kernelization, we repeatedly delete edges incident
	to high-degree vertices.
	Then every vertex has degree at most $dk$. We further remove any vertex or edge that not part of a $K_{1,d}$.
	Now we greedily compute a maximal star matching $M$ and answer 'yes'
	if $M$ has size~$k$.
	Otherwise, we claim that the graph has most $O(k^2)$ edges.
	Since $M$ covers at most $dk$ vertices, the degree bound implies
	that at most $(dk)^2$ edges are incident to~$M$. To count the edges outside $M$, observe that for every such edge $e$, at least one endpoint is adjacent to $M$: otherwise, we know that $e$ is part of a~$K_{1,d}$ (since we haven't removed~$e$) and this $K_{1,d}$ would be disjoint from $M$.
	Each vertex of $G$ outside of~$M$ have at most~$d-1$ neighbors
	outside of $M$ because they would otherwise have been added to~$M$.
	Thus there are at most $(d-1)\cdot (dk)^2$ edges not incident
	to~$M$ and so~$G$ has at most~$d^3\cdot k^2$ edges.
\end{proof}
By Theorem~\ref{thm:Hfactor}, it is unlikely that star matching
problems have kernels with $O(k^{2-\epsilon})$ edges, so the above
kernels are likely to be asymptotically optimal.

\subsection{Packing Paths of Length 3}
Let $P_d$ be the simple path with~$d$ edges. As $P_2$ is the
same as $K_{1,2}$, the problem \matching{P_2} is already covered by
Proposition~\ref{prop:star}, thus we have a $O(k^2)$ upper bound and
a matching $O(k^{2-\epsilon})$ lower bound for this problem. For
\matching{P_3}, the situation is less clear.
Using a similar strategy as in the proof of
Proposition~\ref{prop:star}, it is easy to reduce the maximum degree to
$O(k^2)$ and then argue that the kernel has $O(k^3)$ edges.
Surprisingly, the maximum degree can be further reduced to~$O(k^{1.5})$
using much more complicated combinatorial arguments.
This gives rise to kernels of size~$O(k^{2.5})$ without a tight lower
bound.
\begin{theorem}\label{th:kernel}
	\matching{P_3} has kernels with $O(k^{2.5})$ edges.
\end{theorem}
We prove Theorem~\ref{th:kernel} by showing that the maximum degree~$\Delta$
of the graph can be reduced to~$O(k^{1.5})$.
Once we have an instance~$G$ with maximum degree $\Delta$, we can
obtain a kernel of size~$O(\Delta \cdot k)$ with fairly standard
arguments as follows.
First, we greedily compute a maximal $P_3$-matching.
If we find at least~$k$ paths, then we are done.
Otherwise let~$S$ be the set of at most~$4k$ vertices in the paths.
As every vertex has degree at most~$\Delta$, there are at most~$4 k
	\Delta$ edges incident to~$S$.
Now let us count the number of edges in~$G\setminus S$.
The graph $G\setminus S$ does not contain paths of length $3$, so
every connected component of~$G\setminus S$ is either a triangle or a
star.
Therefore, the average degree is at most $2$ in $G\setminus S$.
If a component of $G\setminus S$ is not adjacent to~$S$, it can be
safely removed without changing the solution.
If a component of $G\setminus S$ has a vertex $v$ with at least two
neighbors in $G\setminus S$ that have degree one in $G$, then we keep
only one of them.
Since every solution uses at most one of them, they are
interchangeable.
After doing this, every component of $G\setminus S$ has at most two
vertices not adjacent to~$S$ in~$G$.
This means that a constant fraction of the vertices in $G\setminus S$
is adjacent to~$S$.
As there are at most $4\Delta k$ edges incident to $S$, this means
that there are at most $O(\Delta k)$ vertices in~$G\setminus S$.
Taking into account that the average degree is at most two
in~$G\setminus S$, we see that there are at most~$O(\Delta\cdot k)$ edges
in~$G\setminus S$.
This yields kernels with $O(k^{2.5})$ edges.
It remains to argue how to reduce the maximum degree to~$\Delta$.

\textbf{Degree reduction.}
Let $G$ be a graph that contains a vertex $v$ with more than $\Delta$
neighbors.
In the following, we call any $P_3$-matching of size $k$ a
\emph{solution}.
Our kernelization procedure will find an edge~$e$ incident to~$v$ that
can be \emph{safely removed}, so that $G$ has a solution if and only
$G\setminus e$ has a solution.
The most basic such reduction is as follows.
\begin{lemma}\label{lem:many2path}
	If there there is a matching $a_1b_1$, $\dots$, $a_nb_n$ of size
	$n\ge 4k+2$ in $G\setminus v$ such that every~$a_i$ is a neighbor
	of $v$, then any single edge $e$ incident to~$v$ can be safely
	removed.
\end{lemma}
\begin{proof}
	Suppose that there is a solution containing a path going
	through~$e$.
	The paths in the solution cover $4k$ vertices, thus without loss of
	generality, we can assume that $a_1$, $b_1$, $a_2$, $b_2$ are not
	used.
	We replace the path containing $e$ with the path $b_1a_1va_2$
	to obtain a solution of $G\setminus e$.
\end{proof}
Let us greedily find a maximal matching $a_1b_1,\dots,a_nb_n$ in
$G\setminus v$ with the requirement that every~$a_i$ is a neighbor
of~$v$.
If $n\ge 4k+2$, we can safely remove an arbitrary edge incident
to $v$ by Lemma~\ref{lem:many2path} and then proceed inductively.
Otherwise, let $M=\{a_1,b_1,\dots,a_n,b_n\}$ be the set of at most
$8k+2$ vertices that are covered by this matching.
Let $X\coloneqq N(v)\setminus M$.
Now every neighbor~$y$ of a vertex $x\in X$ is in $M\cup\{v\}$ since
the matching~$M$ could otherwise have been extended by the edge~$xy$.
In particular, $X$ induces an independent set.
It holds that $|X|\ge 100k$ since otherwise the degree of~$v$ is
smaller than~$\Delta$.

The following technical definition is crucial in our kernelization
algorithm.
\begin{definition}
	Let $u$ be a vertex of $M$ and let $X_u=N(u)\cap X$ be the
	neighborhood of~$u$ in~$X$.
	We call $u$ {\em good} if every set $S\subseteq M$ satisfies
	the following property:
	If there is a matching between $S$ and~$X_u$ of size $|X_u|-1$, then
	$S$ has more than $4k$ neighbors in $X$.
\end{definition}

It is not obvious how to decide in polynomial time whether a vertex is good.
Nevertheless, if certain vertices of~$M$ are known to be good, we can make progress by deleting an edge.
\begin{lemma}\label{lem:move}
	If $x \in X$ has only good neighbors in $M$, then the edge $vx$ can
	be safely removed.
\end{lemma}
\begin{proof}
	We argue that if there is a solution then there is also a solution
	that does not use $vx$.
	If~$vx$ is used as the first or the third edge of a path, the high
	degree of $v$ makes sure that there is a neighbor $y$ of $v$ not used by the
	solution, and we can replace $vx$ by $vy$.
	Now consider a solution that contains a path $P=avxu$ using $vx$ as
	its middle edge; by assumption, $u\in M$ is good.

	By definition, the set $X_u$ contains all vertices $x'$ of $X$ that
	are common neighbors of~$u$ and~$v$.
	Hence, if some vertex $x'\in X_u\setminus x$ is not used by the
	solution, then we can replace~$P$ by $avx'u$.
	Now assume that every $x'\in X_u\setminus x$ is part of some path.
	None of these paths contain $v$.
	If $x'$ is the endpoint of a path, then the {\em mate} of $x'$ is
	its unique neighbor in the path; if~$x'$ is in the middle of a path,
	then the mate of $x'$ is the endpoint that is adjacent to $x'$ in
	the path.
	Recall that every neighbor of $x'\in X$ is in $M\cup\{v\}$, so the
	mate of any~$x'\in X_u\setminus x$ is contained in~$M$.
	The vertices in~$X_u\setminus x$ have distinct mates
	even if two vertices of $X_u\setminus x$ are on the same path.
	This gives rise to a matching between~$X_u\setminus x$ and the
	set~$S\subseteq M$ of all mates of vertices in $X_u\setminus x$.
	Since this matching has size $|X_u|-1$ and $u$ is good, $S$ has at
	least $4k+1$ neighbors in~$X$.
	Thus, some neighbor $y\in X$ of~$S$ is not used by the solution.

	Let $x'\in X_u\setminus x$ be a vertex whose mate $w\in S$ is
	adjacent to a vertex $y\in X$ that is not used by the solution.
	Since~$w$ is the mate of~$x'$, the edge~$wx'$ occurs in a path~$Q$
	of the solution.
	We distinguish two cases.
	If~$x'$ is an endpoint of~$Q$, then we replace the paths~$P=avxu$
	and~$Q=x'wcd$ by the two new paths~$avx'u$ and~$ywcd$.
	If~$x'$ is not an endpoint of~$Q$, then we replace~$P=avxu$
	and~$Q=wx'cd$ by~$ux'cd$ and~$avyw$.
	These are paths since~$x'$ is a common neighbor of~$v$ and~$u$,
	and~$y$ is a common neighbor of~$v$ and~$w$.
	In all cases we found solutions that do not use~$vx$, so~$vx$ can be
	safely removed.
\end{proof}

Next we show that, if $v$ has sufficiently large degree, then we can find a vertex~$x\in X$ for which the the reduction rule in Lemma~\ref{lem:move} is applicable.
\begin{lemma}\label{lem:findgood}
	There is a polynomial-time algorithm that, given a vertex $v$ of
	degree larger than $\Delta=C\cdot k^{1.5}$ for some universal constant $C$, finds a vertex $x\in X$ that has only
	good neighbors in $M$.
\end{lemma}
\begin{proof}
	We maintain an increasing set $M'\subseteq M$ of vertices satisfying the
	invariant that all vertices in $M'$ are good.
	Initially we set $M'=\emptyset$.
	We repeat a procedure that either outputs $x$ as required or adds a
	new good vertex to $M'$.
	If some $x\in X$ does not have neighbors in $M\setminus M'$, then by
	the invariant all neighbors of $x$ in $M$ are good and we can
	output~$x$.
	Otherwise, with $M\setminus M'=\{m_1,\dots,m_t\}$, we can define a
	partition $X^1,\dots,X^t$ of $X$ such that every vertex of $X^i$ is
	adjacent to $m_i$.
	Some of the $X^i$ can be empty.

	We construct a bipartite graph~$H$ that is a subgraph of the
	bipartite graph between~$X$ and~$M$.
	Initially, $H$ has the vertex set $X\cup M$ and no edges.
	We preserve the invariant that every vertex of~$X$ has degree at
	most one in~$H$.

	For every $1\leq i\leq t$ with $|X^i|>1$,
	we add edges to $H$ in the following way.
	For every edge $xy$ of~$G$ with $x\in X^i$ and $y\in M$, let the
	\emph{weight} of $xy$ be the degree $\deg_H(y)$ of $y$ currently in $H$.
	In this weighted graph~$G$, we now compute a matching between~$X^i$
	and~$M$ that has cardinality exactly $|X^i|-1$ and weight at
	most~$4k$.
	This can be done in polynomial time using standard algorithms.
	If there is such a matching, we add all edges of the matching to~$H$
	and continue with the next~$i$.
	This preserves the invariant that every vertex of $X$ has degree at
	most one in~$H$ since the sets $X^i$ are disjoint.
	If there is no such matching, then we claim that $m_i$ is good.
	Assume for contradiction that there is a matching of cardinality
	$|X_{m_i}|-1$ between $X_{m_i}=N(m_i)\cap X$ and a subset
	$S\subseteq M$ that has at most $4k$ neighbors in $X$. Note that this matching matches at least $|X^i|-1$ elements of~$X^i$ to vertices in $S$.
	As~$H$ is a subgraph of $G$, it follows that $S$ has at most $4k$
	neighbors in~$H$.
	This implies that $\sum_{y\in S}\deg_H(y)\le 4k$
	since every vertex of $X$ has degree at most one in~$H$, so the sum
	of the degrees of vertices in $S$ is exactly the size of the
	neighborhood of~$S$ in~$H$.
	This contradicts with the fact that we did not find a suitable
	matching of weight at most~$4k$.
	Thus~$m_i$ is good and can be added to~$M'$.

	We show that unless $|X|=O(k^{1.5})$, the above process finds a good
	vertex in $M$.
	Suppose that the process terminates without finding a good vertex.
	Let $N$ be the number of paths of length two in the final graph $H$
	we obtained.
	As the degree of every vertex of $X$ is at most one in $H$, every
	path of length two is of the form $abc$ with $a\in X^i$, $b\in M$,
	and $c\in X^j$ for some $1\le i , j \le t$.
	Furthermore, we have $i\neq j$: the edges incident to $X^i$ form a
	matching in $H$.
	For some~$i$, let us count the number of paths with $a\in X^i$ and
	$c\in X^j$ for $j<i$.
	Consider a vertex $a\in X^i$ that is not isolated in $H$; it has a
	unique neighbor $b$ in the graph~$H$.
	Consider the graph $H'$ at the step of the algorithm before finding
	the matching incident to $X^i$, and let~$d$ be the degree of $b$
	in~$H'$.
	Then it is clear that $H$ contains exactly $d$ paths of length two
	connecting~$a$ to a vertex of $X^j$ with $j<i$:
	the vertex $b$ has exactly $d$ neighbors in
	$X^1\cup \dots \cup X^{i-1}$.
	Thus if~$S_i$ is the set of vertices that $X^i$ is matched to, then
	the total number of paths between $X^i$ and $\bigcup_{j=1}^{i-1}X^j$
	is exactly the total degree of $S_i$  in $H'$, which is at most $4k$
	by the selection of the matching.
	Thus the total number~$N$ of paths can be bounded by
	$t\cdot 4k\le (8k+2)\cdot 4k=O(k^2)$.

	On the other hand, the number of paths of length two containing
	$m\in M$ as their middle vertex is exactly
	\begin{equation*}
		\binom{\deg_H(m)}{2}\ge \deg_H(m)^2/4-1 .
	\end{equation*}
	Note that $\sum_{m\in M}\deg_H(m)\ge|X|-|M|$: in every nonempty~$X^i$,
	there is exactly one vertex that is isolated in~$H$, and every other
	vertex has degree one.
	Thus the total number of paths is exactly
	\begin{align*}
		\sum_{m\in M} \binom{\deg_H(m)}{2}
		 & \ge \frac{1}{4}\sum_{m\in M}\deg_H(m)^2 -|M|
		\\
		 & \ge \frac{1}{4|M|}\left(\sum_{m\in M}\deg_H(m)\right)^2- |M|
		\\
		 & \ge \frac{1}{4|M|}(|X|-|M|)^2-|M|=\Omega(|X|^2/k)
	\end{align*}
	where we used the relationship between arithmetic and quadratic mean
	in the second inequality, and the facts $|M|\le 8k+2$, $|X|>100k$ in
	the last step.
	Putting together the upper bound $N=O(k^2)$ obtained earlier on the number of path of length 2 in $H$ and the lower bound $N=\Omega(|X|^2/k)$ we have just proved, it follows that
	$|X|=O(k^{1.5})$.

	Thus, we can choose $\Delta=C\cdot k^{1.5}$ for some large enough
	constant $C>0$ so that the above procedure is guaranteed to find a
	vertex $x\in X$ that contains only good neighbors in~$M$.
\end{proof}


\subsection{Packing Paths of Length \texorpdfstring{\boldmath$d$}{d}}
\label{sec:packing-paths-length}

In this section, we study the kernelization of \pp{$P_d$-Matching} which
is given an instance~$(G,k)$ to decide
whether there are at least~$k$ vertex-disjoint paths with~$d$ edges
each.  For $d=1$, this is equal to the problem of finding a standard
matching of size at least~$k$ and therefore computable in polynomial
time.  The results of \textcite{KirkpatrickHell} imply that the
problem is \NP-complete for any constant $d\geq 2$.
Theorem~\ref{thm:setmatching-kernels} implies that the problem has
kernels of size~$O(k^d)$ for every constant~$d$.  Furthermore, we know
kernels of size $O(k^2)$ for $d=2$ (Proposition~\ref{prop:star}) and
$O(k^{2.5})$ for $d=3$ (Theorem~\ref{th:3path}).  The best lower bound
we have is $O(k^{2-\epsilon})$ for every $d\ge 2$.  Surprisingly, we
prove that \pp{$P_d$-Matching} has kernels whose size is bounded by a
polynomial whose degree is independent from~$d$.  We first state a
result about the existence of an annotated kernel of type B2 (in the
notation of Section~\ref{sec:mult-kern}) for \pp{$P_d$-Matching}. With
some additional arguments, it will give a proof of Theorem~\ref{th:uniformpath}.

\begin{theorem}\label{thm:pathpacking}
	For every $d\in\N$, \pp{$P_d$-Matching} has annotated kernels with
	$O\paren[\big]{ d^{d^2} d^7 k^3 }$ vertices.
	The kernelization algorithm runs in time~$\poly(d^{d^2} n)$.
\end{theorem}

\begin{proof}
	If $d < 2$, the problem can be computed in polynomial time, so let $d\ge 2$.

	We describe the kernelization algorithm~A step by step, interleaved with a
	discussion.
	The algorithm is not quite a kernelization algorithm for \pp{$P_d$-Matching},
	but it is one for the following weighted version of the problem:
	The input is a simple graph $G$ with possible self-loops and
	weights~$w_e\in[d]$ on the edges, and a number~$k$.
	We call a self-loop $\{v\}$ a \emph{dangling edge} as this makes more sense in
	the context of paths: We allow paths $e_1,\dots,e_\ell$ to start and end with
	dangling edges, but all other edges have to be regular edges; in the
	degenerate case that the path involves only one vertex, its dangling edge can
	be used only once.
	The goal of the weighted \pp{$P_d$-Matching} problem is to find~$k$
	vertex-disjoint paths such that the weight of each path is at least~$d$.
	The \emph{weight} of a path~$e_1,\dots,e_\ell$ is the sum~$\sum_{i\in[\ell]}
		w_{e_i}$ of its edge weights.
	\pp{$P_d$-Matching} is the restriction of this problem to graphs without
	dangling edges and with weight~$1$ on every edge.

	Let $(G,k)$ be the input for the algorithm.
	\begin{enumerate}[label=(A\arabic*)]
		\item\label{greedy pathpacking}
		      Greedily compute a maximal set of disjoint paths in $G$, each of which has
		      length at most~$d$ and weight at least~$d$.
		      If the path matching contains at least $k$ paths, we have found a solution
		      and are done with the kernelization.
		      Otherwise, let $M\subseteq V(G)$ be the set of vertices used by the paths.
		      We have $|M| \le (d+1) (k-1)$, and the graph $G' \coloneqq G-M$ does not contain
		      any path of weight~$\geq d$.
	\end{enumerate}
	We can find a path of weight~$\ge d$ deterministically in
	time~$\exp\paren[\big]{O(d)} \poly(n)$, so \ref{greedy pathpacking} runs in
	polynomial time.
	Let $G'_1,\dots,G'_t$ be a list of all connected components of $G'$, and
	select arbitrary vertices $r_i \in V(G'_i)$ for each $i\in[t]$.
	We call these vertices \emph{roots}.
	\begin{enumerate}[label=(A\arabic*),resume]
		\item
		      For each $i\in[t]$, compute a depth-first search tree~$T_i$ in $G'_i$
		      starting at the root $r_i$.
		      Let $F\coloneqq \dot\bigcup_{i\in[t]} T_i$ be the corresponding DFS-forest.
	\end{enumerate}
	We have $V(F)=V(G')$.
	The depth of each tree in~$F$ is at most~$d-1$ as otherwise~$M$ would not have
	been maximal.
	For every~$v\in V(F)$, we let $T_v$ be the subtree that consists of~$v$ and
	all of its successors in~$F$.
	Let $u,v\in V(F)$ so that neither $u$ is an ancestor of~$v$ nor vice-versa.
	Then $V(T_u)$ and $V(T_v)$ are disjoint; moreover, if $u$ and $v$ are in the same component of $G'$, then every path in~$G'$ from
	$V(T_u)$ to $V(T_v)$ leads through a common ancestor of~$u$ and~$v$.

	A \emph{request} is a tuple~$(f,i)$ where $f\subseteq V(G)$ with
	$\abs{f}\in\set{1,2}$ prescribes one or two vertices of~$G$, and $i\in[d]$ is
	a weight.
	For a set $S\subseteq V(G)$, we let $\paths^S(f,i)$ be the set of all paths
	in~$G[S\cup f]$ of weight~$\ge i$ such that:
	If $f=\{x,y\}$ for $x\ne y$, then the endpoints of the path are $x$ and $y$,
	and if $f=\{x\}$, then the path starts at~$x$.
	We say that $S$ \emph{satisfies} the request~$(f,i)$ if
	$\paths^S(f,i)\ne\emptyset$.
	Note that $f$ is an edge of~$G$ with~$w_f \ge i$ if and only if the empty set
	(and hence every set~$S$) satisfies~$(f,i)$;
	in this case, we say that the request~$(f,i)$ is \emph{resolved}.

	The idea of the kernelization algorithm is to resolve as many requests as we
	can possibly afford, meaning that we want to increase the weight of edges
	in~$G$ as much as possible.
	Doing so can never destroy existing solutions for the instance~$(G,k)$, but we
	need to carefully avoid introducing solutions in case~$(G,k)$ is a
	no-instance.

	We begin by resolving requests~$(f,i)$ with $f\subseteq M$.
	We let $N_{f,i}$ be the set of all vertices~$v\in V(F)$ such that $V(T_v)$
	satisfies the request $(f,i)$.
	If $u$ is an ancestor of $v$ in~$T$, then $V(T_u)\supseteq V(T_v)$, so all
	requests satisfied by a tree~$T_v$ are also satisfied by the trees~$T_u$ of
	its ancestors.
	Hence $F[N_{f,i}]$ is a subforest of~$F$ such that, for all $j\in[t]$, either
	$V(T_j)$ is disjoint from $N_{f,i}$ or the root $r_j$ of $T_j$ is contained in
	$N_{f,i}$.
	In a rooted tree, a \emph{leaf} is a vertex that does not have any children;
	in particular, the root is a leaf if and only if it is the only node.

	\begin{enumerate}[label=(A\arabic*),resume]
		\item
		      \label{marking}
		      For each unresolved request~$(f,i)$ with $f\subseteq M$:\\
		      If $F[N_{f,i}]$ has more than $dk$ leaves, we set $w_f \coloneqq i$.
	\end{enumerate}

	Here we use the convention that~$w_f=0$ if and only if~$f$ is not an edge, and
	so if $f$ was not an edge before, we add it and set its weight to~$i$.

	To see that \ref{marking} is safe, let $(f,i)$ be an unresolved request such
	that $F[N_{f,i}]$ has more than $dk$ leaves.
	Let~$G$ be the graph before changing the weight of~$w_f$,
	and let~$G^+$ be the graph with the new weight~$w'_f = i$.
	Since $(f,i)$ is unresolved in $G$, we have~${w'_f > w_f}$.
	We need to show that $G$ has $k$ disjoint paths of weight~$\ge d$ if and only
	if $G^+$ does.
	The only if direction is trivial since we only increase weights.
	For the if direction, let $P_1,\dots,P_k$ be disjoint paths of weight~$\ge d$
	in~$G^+$.
	If none of the paths use $f$, the path matching also exists in~$G$.
	Otherwise, suppose without loss of generality that $f$ appears in~$P_1$.
	The set $P= V(P_1)\cup\dots\cup V(P_k)$ has size at most~$dk$.
	Moreover, the sets $V(T_v)$ are mutually disjoint for all leaves~$v$ of
	$F[N_{f,i}]$ since $F$ is a forest.
	Hence the intersection of $P$ and~$V(T_v)$ is empty for at least one leaf~$v$
	of $F[N_{f,i}]$.
	Since $v\in N_{f,i}$, the set $V(T_v)$ satisfies the request~$(f,i)$.
	Hence we can replace $f$ in $P_1$ with a path in~$G$ from
	$\paths^{V(T_v)}(f,i)$, and we obtain a set of $k$ disjoint paths of
	weight~$\ge d$ in the graph~$G$.

	After \ref{marking} has been applied, every request~$(f,i)$ that is still
	unresolved has a corresponding forest~$F[N_{f,i}]$ with at most $dk$ leaves. The number of vertices of the forest~$F[N_{f,i}]$ is at most the number of its leaves times its depth plus 1.
	Since the forest~$F$ does not contain a path of length~$d$, the depth of each
	tree is at most~$d-1$, and the total number of vertices in $F[N_{f,i}]$ is at
	most $d^2k$.
	Let $N'$ be the union of all $N_{f,i}$ for unresolved requests~$(f,i)$
	with~$f\subseteq M$.
	That is,  $N'$ is the set of all vertices $v\in V(F)$ such that $T_v$
	satisfies an unresolved request $(f,i)$ over~$M$.
	There are at most~$\abs{M}^2 \cdot d \le d^3 k^2$ requests over~$M$,
	and~$N_{f,i}$ has size at most~$d^2k$ for each of them, hence the size of~$N'$ is at
	most~$d^5k^3$.
	The algorithm keeps all vertices from $M\cup N'$ in the kernel, so it remains
	to select $\poly(dk)$ vertices from the graph~$G'' \coloneqq G' - N'$.

	Let $C\subseteq V(G'')$ be any connected component of~$G''$.
	We claim that $C$ satisfies the following properties:
	\begin{enumerate}
		\item
		      $C$ does not satisfy any unresolved request over~$M$.
		\item
		      $|N(C)\cap N'| \leq d-1$.
	\end{enumerate}

	For these claims, first note that $C\subseteq V(F)$ and that $G[C]$
	is connected. It is easy to see that $C=V(T_v)$ for some vertex
	$v\in C$. Indeed, let $v$ be a vertex of $C$ whose depth in the DFS
	forest is minimum possible.  The fact that $v\in C$ implies that
	$v\not\in N'$, and it follows from the definition of $N'$ that $T_v$
	is also disjoint from $N'$. Every neighbor of $T_v$ in $G'$ is an
	ancestor $v$, and the ancestors of $v$ are not in $C$ by the minimal
	choice of $v$. Thus $C$ is exactly $V(T_v)$.  To see the first claim, recall that
	the definition of~$v\not\in N'$ is that $V(T_v)$ does not satisfy
	any unresolved request~$(f,i)$ over~$M$.

	For the second claim, since~$F$ is a DFS-forest, the only $G'$-neighbors that
	$V[T_v]$ can have lie on the unique path in~$F$ that leads from a root to $v$.
	In particular, since~${N'\subseteq V(F)}$, all neighbors of $V(T_v)$ that are
	contained in $N'$ must lie on this root-to-$v$ path.
	Finally, any root-to-$v$ path has at most~$d$ vertices including~$v$, which
	implies that~${N(C)\cap N'}$ has size at most $d-1$.

	The neighborhood of each component $C$ is in $M\cup N'$.
	In the rest of the proof, we analyze how a component $C$ can satisfy requests $(f,i)$ with $f\subseteq M\cup N'$ but $f\cap N'\neq\emptyset$.
	Let us give an intuitive reason why we consider only $f\cap N'\neq \emptyset$, that is, ignore requests~$(f,i)$ with $f\subseteq M$.
	A solution to the generalized $P_d$-matching
	instance~$(G,k)$ can use vertices from~$C$ in some number of segments of the paths.
	If such a segment is a path from $\paths^C(f,i)$ for $f\subseteq M$,
	then this segment must correspond to a resolved request~$(f,i)$ since $C$ does
	not satisfy any unresolved requests over~$M$.
	Since the request is resolved, we can replace this segment through~$C$ with
	the edge~$f$ of weight~$\ge i$.
	Thus as we shall see, it can be assumed by a minimality argument that each path segment that
	the solution uses in~$C$ enters through a vertex of~$N'$; they may then end
	within~$C$, exit back out to~$N'$, or exit out to~$M$. Thus it is indeed sufficient to analyse how $C$ can satisfy requests with at least one endpoint in $N'$.
	Furthermore, since there are at most $d-1$ vertices in $N(C)\cap N'$, the solution can
	route at most $d-1$ path segments through~$C$ in total.
	Moreover, all unresolved requests~$(g,i)$ with $g\subseteq M\cup N'$ that $C$
	can satisfy must have $g\subseteq \paren*{N(C)\cap N'} \cup M$ and $g\cap N'
		\ne \emptyset$.
	Therefore, we obtain the following property:
	\begin{enumerate}[resume]
		\item
		      The number of unresolved requests over $M\cup N'$ that~$C$ satisfies is at
		      most $\abs{N(C) \cap N'}\cdot \abs{N(C) \cap (N'\cup M)} \cdot d \le
			      O(d^3k)$.
	\end{enumerate}

	Before reducing the size of each connected component of~$G''$, the algorithm
	first reduces their number by resolving requests~$(g,i)$ over~$M\cup N'$ that
	could be satisfied by~$C$.

	\begin{enumerate}[label=(A\arabic*),resume,start=4]
		\item\label{Nprime marks}
		      For all unresolved requests $(g,i)$ over~$M\cup N'$:
		      If $G''$ has more than $dk$ connected components~$C$ that satisfy the
		      request~$(g,i)$, we set $w_g \coloneqq \max\set{w_g,i}$.
	\end{enumerate}
	The rationale for why these request can be resolved is as in~\ref{marking}:
	Let $P\subseteq V(G)$ be any set of size at most~$dk$ and let $(g,i)$ be a
	newly resolved request.
	Then~$P$ has an empty intersection with at least one component~$C$ that
	satisfies the request~$(g,i)$.
	This means any size-$k$ solution using the edge~$g$ with its new edge weight
	can be rerouted through an available component~$C$ that satisfies the
	request~$(g,i)$.

	We call a component~$C$ \emph{useful} if there is an unresolved request
	$(g,i)$ over $M\cup N'$ that is satisfied by~$C$.
	There are at most $O(d^3k)$ such request.
	Due to \ref{Nprime marks}, each such unresolved request gives rise to at
	most~$dk$ useful components.
	Hence the number of useful components is at most $O\paren*{d^4 k^2}$.
	Let~$N''$ be the set of all vertices that belong to a useful connected
	component of~$G''$.

	\begin{enumerate}[label=(A\arabic*),resume]
		\item
		      Set $G \coloneqq G[M\cup N' \cup N'']$.
	\end{enumerate}

	In other words, this step deletes all components that are not useful.
	To see that doing so is safe, let $G$ be the current graph and let $C$ be a
	component that is not useful.
	The claim is that $G$ has a generalized $P_d$-matching if and only if $G-C$
	does.
	The if direction is trivial since any generalized $P_d$-matching in $G-C$ also
	exists in $G$.
	For the only if direction, let $P_1,\dots,P_k$ be paths in $G$ of weight at
	least~$d$ each, and let them be chosen in such a way that the number of
	vertices of~$C$ that they use is minimized.
	We claim that they actually don't use any vertex of~$C$.

	Assume for contradiction that the solution had a path segment that uses
	vertices of~$C$.
	The path segment cannot have weight~$\ge d$ since $C$ does not contain
	any paths of weight~$\ge d$.
	Hence the path segment has some weight~$i<d$, and it must enter~$C$ from a
	vertex $x$ outside of~$C$, and it may leave to a vertex~$y$ outside of~$C$; in
	particular,~$C$ satisfies the request~$(g,i)$ for some $f\subseteq M\cup N'$.
	The request $(g,i)$ is not resolved, for if it was, we could replace the path
	segment through~$C$ by the edge~$g$ with $w_g\ge i$, which would
	make use of strictly fewer vertices of~$C$.
	Since $g\subseteq N(C)\subseteq M\cup N'$ holds, $(g,i)$ is unresolved, and $(g,i)$
	is satisfied by~$C$, we arrive at a contradiction with the fact that~$C$ is
	not useful.
	Overall, if there is a solution, then there is a solution that does not use
	any vertex of~$C$, and so it is safe to delete~$C$.

	At this point, we are keeping the $\poly(dk)$ vertices of~$M\cup N'$ in the
	kernel as well as $\poly(dk)$ useful components.
	The only issue left to resolve is that useful components may still contain too
	many vertices.
	Recall that at most~$d-1$ segments of paths from a solution can intersect~$C$.
	Since every path has length at most~$d$, every solution~$P$ uses at
	most~$(d-1)(d+1)\le d^2$ vertices of~$C$.
	We will replace~$C$ with a graph~$C'$ of size at most~$f(d)$ such that any set
	of at most $d-1$ requests can be satisfied by~$C$ if and only if it can be
	satisfied by~$C'$.
	We construct~$C'$ as an induced subgraph of~$C$ by using representative sets.

	\begin{enumerate}[label=(A\arabic*),resume]
		\item\label{branching}
		      For each unresolved request~$(g,i)$ over $M\cup N'$ and
		      for each useful connected component~$C$:
		      \begin{enumerate}
			      \item
			            Set $S\coloneqq\emptyset$.
			      \item
			            While $|S|\le d^2$:
			            \begin{enumerate}
				            \item
				                  Find an arbitrary path~$p''$ from $\paths^{C-S}(g,i)$ and mark the
				                  vertices $V(p'')\setminus g$ as useful.
				                  If no such path exists, break the while loop.
				            \item
				                  Branch on $v\in V(p'') \setminus g$ and add $v$ to $S$.
			            \end{enumerate}
		      \end{enumerate}
		\item\label{delete useless vertices}
		      Delete all vertices of useful components that have never been marked as
		      useful in any branch for any request.
	\end{enumerate}

	Before we prove the correctness of \ref{delete useless vertices}, let us first
	analyze the running time and kernel size.  For given $C$ and $(g,i)$,
	\ref{branching} produces a branching tree of depth at most~$d^2$ since~$|S|$
	increases by one in every branching step.
	The fan-out of the branching tree is equal to the size of $V(p)\setminus g$,
	which is at most~$i$ and thus at most~$d$.
	Therefore, the branching tree has at most $O(d^{d^2})$ vertices, and at most
	$O(d^{d^2})$ vertices of~$C$ get marked as useful; this is the number of
	vertices we keep for the component~$C$ and the unresolved request~$(g,i)$.
	Since the number of unresolved requests that~$C$ satisfies is at most $O(d^3
		k)$ and since the number of useful components is at most $O(d^4k^2)$, the
	total number of vertices that remain after~\ref{delete useless vertices} is
	$O(d^{d^2} d^{7} k^3)$.

	For the correctness, let $G$ be the graph before applying~\ref{delete useless
		vertices} and let $G^-$ be the graph after.
	The claim is that $G$ has a solution if and only if $G^-$ does.
	The if direction is trivial since $G^-$ is a subgraph of~$G$.
	For the only if direction, let $P_1,\dots,P_k$ be disjoint paths of~$G$, each
	of weight~$\ge d$, and let this solution be chosen as to minimize the number $u$ of useless vertices being used and, subject to that, to minimize the number $o$ of other vertices of~$G''$ that are being used by the solution.
	If $u=0$ zero, then $P_1,\dots,P_k$ is also a solution of~$G^-$, what we wanted to show.
	Assume for contradiction that $u>0$ and the solution uses a useless vertex~$v$, and let
	$C$ be the useful connected component of~$G''$ that contains~$v$.

	Let $(g,i)$ for $g\subseteq M\cup N'$ be the request to~$C$ that the solution
	$P_1,\dots,P_k$ satisfies with the path segment $p=e_1,\dots,e_\ell$ that
	uses~$v\in C$.
	That is, $p$ is a connected subgraph of some~$P_j$ and satisfies the following
	properties:
	\begin{enumerate}
		\item
		      $p\in\paths^C(g,i)$ and
		\item
		      if $p$ contains an endpoint of~$P_j$, then $\abs{g}=1$; otherwise
		      $\abs{g}=2$.
	\end{enumerate}
	We remark that a path segment~$p$ cannot contain both endpoints of $P_j$
	since~$C$ does not contain paths of weight~$\ge d$.
	If~$(g,i)$ was resolved, we could replace the corresponding path segment by
	the edge $g$ of weight~$w_g\geq i$, which would lead to a solution with fewer
	useless vertices or fewer useful vertices of~$G''$.
	Since we minimized both these numbers, the request~$(g,i)$ is unresolved and
	thus considered in \ref{branching}.

	We now show that~$p$ can be replaced by another path~$p' \in \paths^C(g,i)$ so
	that $p'$ only uses vertices of~$C$ that were marked as useful; this then
	contradicts the fact that the solution minimized the number of useless
	vertices.
	Thus the solution uses only useful vertices.
	Let $T\subseteq C$ be the set of all vertices of~$C$ that are used by the
	solution, but not including the vertices of the segment~$p$.
	That is, $T \coloneqq \paren[\Big]{ \paren*{V(P_1)\cup\dots\cup V(P_k)} \setminus
			V(p) } \cap C$.
	By previous considerations, the size of $T$ is at most~$d^2$.
	Clearly, $\paths^{C\setminus T}(g,i)$ contains~$p$, and we prove that it also
	contains a path~$p'$ that only uses useful vertices.
	Consider the branching tree that \ref{branching} builds for the
	request~$(g,i)$ and the component~$C$.
	We start at the root of the branching tree and we follow a single branch in
	our argument; we only follow a branch if it satisfies~$S\subseteq T$.
	In step b) ii) of such a branch, the set $\paths^{C\setminus S}(g,i)$ is not
	empty, for it contains the path~$p$.
	Let $p''$ be the path that ends up being selected and whose vertex
	set~$V(p'')$ gets marked as useful.
	We either have $V(p'')\cap T=\emptyset$, and so we have found our
	path~$p'=p''$, or we have $V(p'') \cap T \neq \emptyset$.
	In the latter case, we follow a branch in which~$v\in V(p'')\cap T$ gets added
	to~$S$; this clearly preserved the invariant~$S\subseteq T$.
	The process stops at the latest when $S=T$ holds, in which case $V(p'')\cap
		T=\emptyset$ is guaranteed and $p'=p''$ is the path we want to find.
	Overall, we replaced the segment~$p$ with a segment~$p'$ that uses only useful
	vertices, a contradiction.
	Hence the solution $P_1,\dots,P_k$ only uses useful vertices, which means it's
	also a solution of~$G^-$.
\end{proof}
Note that \pp{$P_d$-Matching} is \NP-hard, and in particular, the reduction takes
time $f(d) \cdot \poly(n)$.
Thus we can compose the compression above with the \NP-hardness reduction to
obtain kernels of size $f(d)\cdot \poly(k)$ for \pp{$P_d$-Matching}.
We can also obtain a kernel more directly by modifying the argument above as
follows: Every time we increase the weight of an edge~$g$ in step (A3) or (A4), we store a set $V_g$
of at most $d^2k$ vertices of~$G'$ so that $\paths^{V_g}(g,w_g)$ contains at
least $dk+1$ paths that are internally disjoint.
These vertices witness that we were allowed to increase the weight to~$w_g$.
In the end, we would also keep the vertex sets~$V_g$ in the kernel.
This does not change the asymptotics for the number of vertices in the kernel.

Following the notation introduced in Section~\ref{sec:mult-kern}, Theorem~\ref{thm:pathpacking} and the arguments in the previous paragraph show that \pp{$P_d$-Matching} has uniform kernels of type B2. By Lemma~\ref{lem:B2toA2B1}, it follows that \pp{$P_d$-Matching} has kernels of type A2 and B1 as well, proving Theorem~\ref{th:uniformpath} stated in the the introduction.

As \pp{$P_d$-Matching} has kernels of type A2 and B1, it is natural to wonder whether
it has kernels of type A1.
Restricting the problem to $k=1$ yields the $d$-Path problem.
An A1-kernelization for \pp{$P_d$-Matching} would imply a polynomial kernel for
$d$-Path, which it does not have unless \hypofails.

\section*{Acknowledgements}
We thank Martin Grohe and Dieter van Melkebeek for
valuable comments on previous versions of this paper.

\printbibliography

\appendix
\section{Sunflower Kernelization for Set Matching}
\label{appendix sunflower setmatching}

We sketch a modern proof of Theorem~\ref{thm:setmatching-kernels},
that  \setmatching{d} has kernels with $O(k^d)$ hyperedges.
\begin{proof}[Sketch]
	A \emph{sunflower} with $p$ petals is a set of $p$ hyperedges whose
	pairwise intersections are equal.
	By the sunflower lemma, any $d$-uniform hypergraph~$G$ with more than
	$d!\cdot r^d$ edges has a sunflower with~$r+1$ petals~\parencite{ER60sunflower} and in fact such a sunflower can be found in linear time.
	We set $r=dk$ and observe that, in any sunflower with $r+1$ petals, we
	can arbitrarily choose an edge $e$ of the sunflower and remove it from
	the graph.
	To see this, assume we have a matching $M$ of $G$ with~$k$ edges.
	If~$M$ does not contain $e$, then $M$ is still a matching of size~$k$
	in $G-e$.
	On the other hand, if~$M$ contains $e$, there must be a petal that
	does not intersect~$M$ since we have $dk+1$ petals but~$M$ involves
	only~$dk$ vertices.
	Thus we can replace $e$ in the matching by the edge that corresponds
	to that petal, and we obtain a matching of $G'$ that consists of~$k$
	hyperedges.
	This establishes the completeness of the reduction.
	The soundness is clear since any matching of $G'$ is a matching
	of~$G$.
\end{proof}

\section{Multicolored Biclique}
\label{appendix:multicolor-biclique}

\begin{lemma}
	\pp{Multicolored Biclique} is \NP-complete.
\end{lemma}
\begin{proof}
	Let graph $G$ and integer $k$ be an instance of \pp{Clique}. Let
	$\{v_i\mid 1\le i \le n\}$ be the vertex set of $G$. We construct
	a bipartite graph $B$ on vertex set $\{u_{i,j},w_{i,j}\mid 1\le i \le
		k, 1\le j \le n\}$. We make vertices $u_{i,j}$ and $v_{i',j'}$ adjacent if and only if
	\begin{itemize}
		\item either $(i,j)=(i',j')$ or
		\item  $i\neq i'$ and vertices $v_j$ and $v_{j'}$ are adjacent.
	\end{itemize}

	Consider the partitions $U=U_1\cup\dots\cup U_k$ and
	$W=W_1\cup \dots\cup W_k$,
	where $U_i=\{u_{i,j}\mid 1\le j \le n\}$
	and $W_i=\{w_{i,j}\mid 1\le j \le n\}$.
	We claim that $B$  contains a biclique $K_{n,n}$
	respecting these partitions if and only if $G$ contains a $k$-clique.
	It is easy
	to see that if $\{v_{a_1}, \dots, v_{a_k}\}$ is a clique in~$G$, then
	$\{u_{1,a_1},\dots,u_{k,a_k},w_{1,a_1},\dots,w_{k,a_k}\}$ is a biclique of
	the required form in $B$. On the other hand, if
	$\{u_{1,a_1},\dots,u_{k,a_k},w_{1,b_1},\dots,w_{k,b_k}\}$ is such a
	biclique, then
  $a_i=b_i$ for every $1\le i\le k$; otherwise~$u_{i,a_i}$ and~$w_{i,b_i}$
  are not adjacent.
	It follows that $\{v_{a_1},\dots,v_{a_k}\}$ is a clique in $G$:
	if~$v_{a_i}$ and~$v_{a_{i'}}$ are not adjacent in $G$ (including the
	possibility that $a_i=a_{i'}$), then $u_{i,a_i}$ and
	$w_{i',b_{i'}}=w_{i',a_{i'}}$ are not adjacent in $B$.
\end{proof}

\section{Lower Bounds for Vertex Cover in Hypergraphs}
\label{appendix:dVC}
We present an elementary reduction from $\OR(\pp{$3$-Sat})$ to
\pp{$d$-Vertex Cover}, i.e., the vertex cover problem in $d$-uniform
hypergraphs.
The $d$-partiteness flavor is crucial in the reduction, but it is not
necessary to explicitly spell out the $d$-partite problem $L$ as we did with \pp{Multicolored Biclique} before.

Recall that a subset $S$ of vertices in a $d$-partite hypergraph is a clique if all of the~$\binom{|S|}{d}$ size-$d$ subsets of $S$ are edges of the hypergraph. Similarly, $S$ is an independent set if none of these $\binom{|S|}{d}$ sets is an edge of the hypergraph. The set $S$ is a vertex cover if every edge of the hypergraph has a non-empty intersection with~$S$. The complement $\overline{H}$ of a $d$-uniform hypergraph contains edge $e$ if and only if $H$ {\em does not} contain $e$. Analogously to graphs, $H$ has an independent set of size $k$ if and only if $\overline{H}$ has a clique of size $k$. Moreover, an $n$-vertex hypergraph has an independent set of size $k$ if and only if it has a vertex over of size $n-k$.

\begin{theorem}[\cite{DellVanMelkebeek}]
  \label{thm:dvertexcover}
	Let $d\geq 2$ be an integer.
  If $\hypo$ holds, then \pp{$d$-Vertex Cover} does not have
  kernels of size~$O(k^{d-\epsilon})$.
\end{theorem}
\begin{proof}
	Let $\varphi_1,\dots,\varphi_t$ be $t$ instances of \pp{$3$-Sat},
	each of size $s$.
	Without loss of generality, assume that the set of variables
	occurring in the formulas is a subset of $[s]$.
	Let~$P$ be the consistency graph on partial assignments that assign
	exactly three variables of~$[s]$.
	More precisely, the vertex set of $P$ is the set of functions
	$\sigma:S\to\{0,1\}$ for sets $S\in\binom{[s]}{3}$,
	and two partial assignment $\sigma,\sigma' \in V(P)$ are adjacent
	in~$P$ if and only if $\sigma$ and $\sigma'$ are consistent, i.e.,
	they agree on the intersection of their domains.
	Now the cliques of size~$\binom{s}{3}$ in $P$ are exactly the
	cliques that are obtained from full assignments $[s]\to\{0,1\}$ by
	restriction to their three-variable sub-assignments.
	Furthermore, $P$ has no clique of size larger than~$\binom{s}{3}$.

	We construct a $d$-uniform hypergraph $G$ on $n=t^{1/d+o(1)}\poly(s)$ vertices that has a clique with some number $k$ of vertices if and only if
	some~$\varphi_i$ is satisfiable. We known that this is further equivalent to $\overline{G}$ having an independent set of size $k$, or to $\overline{G}$ having a vertex cover of size $n-k=t^{1/d+o(1)}\poly(s)$.  Thus  such a reduction is sufficient to apply  Lemma~\ref{lemma:kernel-proxy} and to prove the lower bound for  \pp{$d$-Vertex Cover}.

	We use a suitable bijection between $[t]$ and $[t^{1/d}]^d$, and we
	write the $\varphi_i$'s as $\varphi_{b_1,\dots,b_d}$ for
	$(b_1,\dots,b_d)\in[t^{1/d}]^d$.
	The vertex set of $G$ consists of $d\cdot t^{1/d}$ groups of
	vertices $V_{a,b}$ for $a\in[d]$ and $b\in [t^{1/d}]$.
	We consider each set $V_{1,b}$ as a copy of the vertex set of $P$,
	and for $a>1$, we let $|V_{a,b}|=1$ for all $b$.
	A subset $e$ of $d$ elements of $V(G)$ is a hyperedge in~$G$ if and
	only if the following properties hold:
	\begin{enumerate}
		\item each $a\in [d]$ has at most one $b=b_a\in [t^{1/d}]$ for
		      which $e\cap V_{a,b} \neq \emptyset$,
		\item $e\cap V_{1,b_1}$ corresponds to a clique in~$P$, and
		\item if $e\cap V_{a,b_a}\neq\emptyset$ for all $a$, then the
		      (unique) partial assignment $\sigma\in e\cap V_1$ does not set any clause
		      of $\varphi_{b_1,\dots,b_d}$ to false.
	\end{enumerate}
	Edges with $|e\cap V_{1,b_1}|>1$ play the role of checking the
	consistency of partial assignments, and edges with
	$|e\cap V_{a,b_a}|=1$ for all $a$ select an instance~$\varphi_b$ and
	check whether that instance is satisfiable.

	We set $k=\binom{s}{3}+d-1$.
	For the completeness of the reduction, let $\sigma:[s]\to\{0,1\}$ be
	a satisfying assignment of $\varphi_{b_1,\dots,b_d}$.
	Let $C$ be the set of all three-variable sub-assignments of~$\sigma$
	in the set $V_{1.b_1}$, and we also add the $d-1$ vertices of
	$V_{2,b_2}\cup\dots\cup V_{d,b_d}$ to $C$.
	We claim that~$C$ induces a clique in $G$.
	Let $e$ be a $d$-element subset of $C$, we show that it is a
	hyperedge in $G$ since it satisfies the three conditions above.
	Clearly, $e\subseteq C$ is fully contained in
	$V_{1,b_1}\cup\dots\cup V_{d,b_d}$ and satisfies the first
	condition.
	The second condition is satisfied since $e\cap V_{1,b_1}$ contains
	only sub-assignments of the full assignment $\sigma$.
	The third condition holds since $\sigma$ is a satisfying
	assignment and therefore none of its sub-assignments sets any clause
	to false.

	For the soundness, let $C$ be a clique of size $k$ in $G$.
	By the first property, $C$ intersects at most one set $V_{a,b_a}$
	for all $a$.
	Also, the intersection $C\cap V_{1,b_1}$ induces a clique in $G$ and
	therefore corresponds to a clique of $P$.
	By the properties of $P$, this intersection can have size at most
	$\binom{s}{3}$, and the only other vertices $C$ can contain are the
	$d-1$ vertices of $V_{2,b_2}\cup\dots\cup V_{d,b_d}$.
	Thus, we indeed have $|C\cap V_{1,b_1}|=\binom{s}{3}$ and
	$|C\cap V_{2,b_2}|=\dots=|C\cap V_{d,b_d}|=1$.
	The properties of $P$ imply that the first intersection corresponds
	to some full assignment $\sigma:[s]\to\{0,1\}$.
	By the third property, no three-variable sub-assignment sets any
	clause of $\varphi_{b_1,\dots,b_d}$ to false, so $\sigma$
	satisfies the formula.

	Thus, $(G,k) \in \pp{$d$-Clique}$ if and only if
	$(\varphi_1,\dots,\varphi_t) \in \OR(\pp{$3$-Sat})$. Since $G$
	and $k$ are computable in time polynomial in the bitlength of
	$(\varphi_1,\dots,\varphi_t)$ and
	$n\coloneqq |V(G)|\leq t^{1/d} \cdot \poly(s)$,
	we have established the polynomial-time mapping reductions that are required to apply
  Lemma~\ref{lemma:kernel-proxy}.
  Thus, if $\hypo$ holds, \pp{$d$-Clique} does not have kernels of size~$O(n^{d-\epsilon})$.
Due the bijection mentioned in the paragraph before the statement of Theorem~\ref{thm:dvertexcover}, between $k$-cliques in~$G$, $k$-independent sets in the complement hypergraph~$\overline{G}$, and $(n-k)$-vertex-covers in~$\overline{G}$, the claim follows, namely that \pp{$d$-Vertex-Cover} does not have kernels of size~$O(k^{d-\epsilon})\le O(n^{d-\epsilon})$ if \hypo{}.
\end{proof}
\end{document}